\documentclass{article}
\usepackage[margin=1in]{geometry}
\usepackage[english]{babel}
\usepackage[utf8]{inputenc}
\usepackage{subcaption}
\usepackage{amsmath}
\usepackage{mathtools}
\usepackage{amssymb}
\usepackage{amsfonts}
\usepackage{amsthm}
\usepackage{mathrsfs}
\usepackage[all]{xy}
\usepackage{graphicx}
\usepackage{color}
\usepackage{cite}
\usepackage{url}
\usepackage[labelfont=bf,labelsep=period,justification=raggedright]{caption}
\usepackage[english]{babel}
\usepackage[utf8]{inputenc}
\usepackage{hyperref}
\usepackage[colorinlistoftodos]{todonotes}
\usepackage{tkz-fct}
\usepackage{tikz}
\usepackage{authblk}
\usetikzlibrary{calc}

\title{Non-Commutativity in Geodesic Congruences}
\author[1]{Aarav Shah\thanks{shahaarav103@zohomail.in}}
\author[2]{Louis H. Kauffman \thanks{loukau@gmail.com}}

\affil[1]{Virtual Institute of Astroparticle Physics, 75018 Paris, France}

\affil[2]{Department of Mathematics, Statistics and Computer Science University of Illinois at Chicago, 851 South Morgan Street, Chicago, IL, 60607-7045 
}

\date{\today}

\theoremstyle{plain}
\newtheorem{thm}{Theorem}
\newtheorem{lemma}[thm]{Lemma}

\newtheorem{prop}[thm]{Proposition}

\theoremstyle{definition}
\newtheorem{defn}[thm]{Definition}

\theoremstyle{remark}

\numberwithin{equation}{section}
\numberwithin{thm}{section}

\begin{document}

\maketitle
\begin{abstract}
   We propose a geometric realization of spacetime non-commutativity where the spacetime manifold remains smooth while the commutation relation between tangent and separation vectors of geodesic congruences is generalized via Frobenius' theorem. This non-commutative kinematic deformation leads to a modified Raychaudhuri equation with a shifted expansion scalar $\tilde \theta=\theta+\theta^{\rm NC}$. We introduce a timelike and null-like Non-Commutative Energy Condition (NCEC) under which focusing and geodesic incompleteness can, in some regimes, persist even when the standard strong energy condition is violated. Using this, we generalize the key focusing lemmas, conjugate-point propositions, and singularity theorems for both timelike geodesics orthogonal to a Cauchy surface and null geodesics orthogonal to a trapped surface. Applied to a FLRW spacetime, we argue that $\theta^{\rm NC}$ modifies both the Friedmann acceleration equation and the continuity equations. The non-commutative corrections separate into a universal modification of the dilution of the standard cosmological fluids and another contribution to the dark energy sector that is non-linear in energy density. Crucially, this provides a purely geometric origin for the evolution of dark energy, without introducing exotic dark energy fields or modifying Einstein's field equations.
\end{abstract}
\section{Introduction}
Understanding the microscopic structure of spacetime remains one of the central open problems in modern theoretical physics \cite{Carlip_2001,rovelli1998stringsloopsotherscritical,Blau:2009fzi}. While General Relativity \cite{Wald1984,HawkingEllis1973,oneill1983semiriemannian,joshi1993global,MisnerThorneWheeler1973,carroll2003spacetime} provides an exceptionally successful classical description of gravitation and the large-scale dynamics of the Universe, it is widely expected to break down at sufficiently high energies, where quantum effects become significant \cite{Carlip_2001,rovelli1998stringsloopsotherscritical,Donoghue_1994}. One of the recurring themes emerging from attempts to formulate a theory of quantum gravity is that the fundamental geometry of space time may require modification \cite{Carlip_2001,rovelli1998stringsloopsotherscritical}. Among the various possibilities that have been explored, spacetime non-commutativity, whose first published Lorentz-invariant construction dates back to Snyder \cite{Snyder:1946qz}\footnote{The underlying idea predates Snyder: in a 1930 letter to Rudolf Peierls, Heisenberg considered noncommuting position coordinates as a possible way of addressing short-distance singularities. The idea subsequently circulated through Peierls and Pauli to Oppenheimer and ultimately to his student Hartland Snyder, who provided its first published realization.}, has attracted considerable attention, motivated by developments in string theory \cite{Seiberg_1999,Ishibashi_1997}, non-commutative quantum field theory \cite{Douglas_2001,Szabo_2003}, deformation quantization \cite{Bayen:1977ha,Bayen:1977hb,10.4310/jdg/1214455536}, matrix models \cite{Connes1994,steinacker2013noncommutativegeometrymatrixmodels}, and non-commutative geometry \cite{Connes1994,Chamseddine_1997,Chamseddine_2007}.
\\
\\
In most existing approaches, however, spacetime non-commutativity is introduced by replacing the algebra of smooth functions on the spacetime manifold with a non-commutative algebra \cite{Connes1994,Bayen:1977ha,Bayen:1977hb,Douglas_2001,Szabo_2003}. In formulations based upon Moyal deformations \cite{Groenewold:1946kp,Moyal1949,Bayen:1977ha,Bayen:1977hb,Douglas_2001,Szabo_2003} or Connes' non-commutative geometry \cite{Connes1994,Chamseddine_1997,Chamseddine_2007}, the emphasis is therefore placed upon the non-commutativity of the spacetime coordinates themselves. Although these approaches have provided deep mathematical insights and have found numerous applications in high-energy physics \cite{Connes1994,Seiberg_1999,Douglas_2001,Szabo_2003}, they also replace, or substantially modify, the conventional geometric description of spacetime, making the connection with classical gravitational physics less direct.
\\
\\
Nevertheless, this need not be the only realization of spacetime non-commutativity. More generally, one may ask whether non-commutative effects can emerge while preserving the smooth differentiable manifold \cite{lee2012introduction,lee2018introduction,doCarmo1992Riemannian,spivak1999comprehensive,warner1983foundations} of General Relativity \cite{Wald1984,HawkingEllis1973,oneill1983semiriemannian}. Rather than modifying the underlying spacetime itself, it is natural to consider the possibility that the non-commutative structure is encoded in the kinematics of geodesic congruences \footnote{A related recent work \cite{Matt_Gregory_2026} has explored a field-theoretic perspective to the geodesic equation in a non-commutative space. However, it adopts a different underlying realization of non-commutativity and has different aims.}. Such a framework retains the standard differential-geometric setting while relaxing the assumption that the tangent vector field of a geodesic congruence commutes with the separation vector connecting neighboring geodesics. As follows from the Frobenius theorem \cite{frobenius1877ueber,warner1983foundations}, non-vanishing commutators between these vector fields arise naturally whenever the associated two-dimensional distribution is not integrable.
\\
\\
This relaxation of the commutativity assumption has direct consequences for the classical singularity theorems, which constitute one of the principal achievements of mathematical General Relativity. The Penrose-Hawking singularity theorems heavily rely on the focusing behavior of the expansion $\theta$ under the Raychaudhuri equation \cite{Raychaudhuri1957SingularSI,1955PhRv...98.1123R,HawkingEllis1973,Witten_2020}, given the strong energy condition on the Ricci tensor and a suitable initial convergence of a geodesic congruence, the equation forces $\theta\rightarrow -\infty$ within a finite affine or proper time, signaling the development of a conjugate point \cite{PhysRevLett.14.57,HawkingEllis1973,Witten_2020}. Combined with the global causal structure of spacetime, the existence of such conjugate points leads to geodesic incompleteness under the hypotheses of the singularity theorems \cite{PhysRevLett.14.57,PhysRevLett.17.444,Hawking:1970zqf,HawkingEllis1973,Witten_2020,Senovilla_2022}. Since this argument depends sensitively on the precise form of the Raychaudhuri equation and on the energy condition invoked, any modification to the underlying congruence kinematics necessarily alters the conditions under which geodesic incompleteness can be established. In the framework developed here, the physically relevant expansion is no longer $\theta$ itself but the shifted quantity $\tilde \theta=\theta+\theta^{\rm NC}$, and the Raychaudhuri equation for $\tilde \theta$ acquires additional terms proportional to $\theta^{\rm NC}$ and its proper-time derivative. We show that the ordinary strong energy conditions remain sufficient to guarantee focusing of $\tilde \theta$ provided $\theta^{\rm NC}$ stays finite along the geodesic up to the would-be focusing point, however such a stirng condition is no longer necessary in certian regimes and we can obtain weaker energy conditions due to the presence of non-commutativity. To capture the more general case we introduce a timelike and null-like non-commutative energy condition (NCEC), Crucially, in certain regimes, the NCEC can be satisfied even when the strong energy condition is violated outright, so that focusing, and consequently geodesic incompleteness, can persist in regimes where the classical Penrose-Hawking  \cite{Hawking:1967ju,Hawking:1970zqf,HawkingEllis1973} wouldn't apply. We use this modified criterion to establish non-commutative generalizations of the principal conjugate-point lemmas \cite{HawkingEllis1973,Witten_2020}, the extremal-length characterization of geodesics \cite{Wald1984,HawkingEllis1973,oneill1983semiriemannian,joshi1993global}, and ultimately of the two central singularity theorems: incompleteness of past-directed timelike geodesics orthogonal to a Cauchy surface of sufficiently negative extrinsic curvature \cite{Hawking:1967ju}, and incompleteness of null geodesics orthogonal to a trapped surface in a globally hyperbolic spacetime \cite{PhysRevLett.14.57}. In both cases, the logical structure of the classical proof is preserved essentially unchanged. 
\\
\\
Non-Commutativity of geodesic congruences is also attractive from a cosmological standpoint \cite{mukhanov2005physical,2003moco.book.....D,weinberg2008cosmology}. The standard cosmological model successfully describes a broad range of observations \cite{2020}, yet fundamental questions remain unresolved, including the physical origin of dark energy \cite{Perlmutter_1999,Riess_1998,Frieman_2008,Huterer_2017,Joyce_2015}, the Hubble tension \cite{Perivolaropoulos_2022,Di_Valentino_2021}, and the possibility that additional geometric effects become relevant on cosmological scales. Intriguingly, recent observational analyses, such as those from the Dark Energy Spectroscopic Instrument (DESI) \cite{Adame_2025,Gialamas_2025,Giar_2025} hint at a dynamical dark energy \cite{Ratra:1987rm,Wetterich_1988,Caldwell_2002,COPELAND_2006} equation of state over a static cosmological constant, renewing the quest for a natural mechanism behind dynamical dark energy. 
Since the Raychaudhuri equation \cite{1955PhRv...98.1123R,HawkingEllis1973} governs the evolution of geodesic congruences and therefore underlies the dynamical evolution of homogeneous and isotropic cosmologies, modifying its geometric foundations provides a natural avenue for investigating whether non-commutativity alone can generate new cosmological phenomena without introducing additional exotic matter scalar fields\footnote{Here, by the term “matter scalar field”, ” we refer to a scalar field representing an independent physical quantity and degree of freedom, distinct from the underlying geometric structure of spacetime.} or modifying Einstein's field equations. In contrast, previous field-theoretic approaches deriving geodesic equations for  Moyal-type non-commutative spacetimes, such as that of the authors of \cite{Matt_Gregory_2026}, do not lead to significant modifications of the underlying cosmological dynamics in an FLRW spacetime, at least at leading order.
\\
\\
The remainder of this paper is organized as follows. In Sec.\ref{Section 2} we review the standard construction of geodesic congruences, the derivation of the Raychaudhuri equation for both timelike and null congruences, and the classical machinery of conjugate points, extremal length, and the Penrose-Hawking singularity theorems, establishing the notation and results that will be generalized in the sections that follow. In Sec. \ref{Section 3} we introduce the non-commutative deformation of the congruence kinematics motivated by the Frobenius theorem \cite{frobenius1877ueber,warner1983foundations}, derive the resulting effective deformation tensor and modified Raychaudhuri equation, and construct the NCEC. We use this to establish generalized focusing lemmas, conjugate-point propositions, and extremal-length theorems, culminating in non-commutative versions of the singularity theorems for both timelike geodesics orthogonal to a Cauchy surface and null geodesics orthogonal to a trapped surface. In Sec.\ref{Section 4} we specialize the modified Raychaudhuri equation to a homogeneous and isotropic FLRW spacetime, derive the corresponding modification to the continuity equations for matter, radiation, and curvature, and identify another contribution that is non-linear in energy density, with the dark-energy sector. We then conclude
our work in Sec.\ref{Section 5}.
\section{Singularities}\label{Section 2}
This section reviews standard results on geodesic congruences and the singularity theorems, following Wald \cite{Wald1984} as the main reference.
\\
\\
Let $\mathcal{M}$ be a manifold and let $O\subset \mathcal{M}$ be open. A congruence in $O$ is a family of curves such that through each $p\in O$ there passes precisely one curve in this family. Thus, the tangents to a congruence yield a vector field in $O$, and, conversely, every continuous vector field generates a congruence of curves. The congruence is said to be smooth if the corresponding vector field is smooth. 
\\
\\
Consider, first, a smooth congruence of timelike geodesics \cite{Wald1984,ellis2012relativistic}. Without loss of generality, we may assume that the geodesics are parametrized by proper time $\tau$, so that the vector field $\xi^a$, of tangents, is normalized to unit length, $U^\mu U_\mu=-1$. Then, the tensor field, $B_{\mu\nu}$, defined by
\begin{equation}
    B_{\mu\nu}=\nabla_{\nu}U_{\mu},
\end{equation}
will be purely spatial, i.e,
\begin{equation}
    B_{\mu\nu}U^\mu=B_{\mu\nu}U^\nu=0.
\end{equation}
The physical interpretation of $B_{\mu\nu}$ may be seen from the following considerations. Consider a smooth one-parameter subfamily $\gamma_s(\tau)$ of geodesics in the congruence, and let $V^a$ be the orthogonal deviation vector from $\gamma_0$ for this subfamily. Then $V^\mu$ represents an infinitesimal spatial displacement from $\gamma_0$ to a nearby geodesic in the subfamily. In the conventional formalism, the tangent and deviation vectors commute. In other words, the lie derivative \cite{schutz1980geometrical,lovelock1989tensors,lee2012introduction} of the deviation vector with respect to the tangent vector vanishes,
\begin{equation}
    \mathcal{L}_{U}V^\mu=[U,V]^{\mu}=0.
\end{equation}
Thus,
\begin{equation}
    U^\mu\nabla_\mu V^\nu=V^\nu\nabla_\mu U^\mu=B^\nu_\mu V^\mu,
\end{equation}
and hence, $B^\mu_\nu$ measures the failure of $V^\nu$ to be parallelly transported. An observer on the geodesic $\gamma_0$ would find the nearby geodesics surrounding him to be stretched and rotated by the linear map $B^a_b$.
\\
\\
We define the spatial metric $P_{\mu\nu}$ by
\begin{equation}
    P_{\mu\nu}=g_{\mu\nu}+U_aU_b.
\end{equation}
Thus, $P^\mu_\nu=g^{\mu\rho}h_{\rho \nu}$ is the projection operator onto the subspace of the tangent space perpendicular to $U^a$. We define the expansion $\theta$, shear $\sigma_{\mu\nu}$, and twist $\omega_{\mu\nu}$ of the congruence by
\begin{equation}
    \theta=B^{\mu\nu}h_{\mu\nu},
\end{equation}
\begin{equation}
    \sigma_{\mu\nu}=B_{(\mu\nu)}-\frac{1}{3}\theta P_{\mu\nu},
\end{equation}
\begin{equation}
    \omega_{\mu\nu}=B_{[\mu\nu]}.
\end{equation}
Thus, $B_{\mu\nu}$ is decomposed as
\begin{equation}
    B_{\mu\nu}=\frac{1}{3}\theta P_{\mu\nu}+\sigma_{\mu\nu}+\omega_{\mu\nu}.
\end{equation}
It is clear that the geodesic equation,
\begin{equation}
    U^{\nu}\nabla_{\nu}U^\mu=0,
\end{equation}
will yield equations for the rate of change of $\theta,\sigma_{\mu\nu},\omega_{\mu\nu}$. We have
\begin{equation}\label{geodesic deviaiton equation}
\begin{split}
U^{\rho}\nabla_{\rho}B_{\mu\nu}=U^\rho\nabla_\rho\nabla_\nu U_\mu=U^\rho\nabla_\nu\nabla_\rho U_\mu+R_{\rho\nu\mu}^{\sigma}U^{\rho}U_{\sigma}\\
=\nabla_{\nu}(U^\rho\nabla_\rho U_\mu)-(\nabla_\nu U^\rho)(\nabla_\rho U_\mu)+R_{\rho\nu\mu}^{\sigma}U^{\rho}U_{\sigma}\\
=-B^\rho_\nu B_{\mu \rho}+R_{\rho\nu\mu}^{\sigma}U^{\rho}U_{\sigma}.
\end{split}
\end{equation}
Taking the trace of eqn.\eqref{geodesic deviaiton equation}, we obtain
\begin{equation}\label{Raychaudhuri's equation}
    U^\rho \nabla_\rho \theta=\frac{d\theta}{d\tau}=-\frac{1}{3}\theta^2-\sigma_{\mu\nu}\sigma^{\mu\nu}+\omega_{\mu\nu}\omega^{\mu\nu}-R_{\rho\sigma}U^\rho U^\sigma .
\end{equation}
The trace-free, symmetric part of eqn.\eqref{geodesic deviaiton equation} yields
\begin{equation}
   U^\rho \nabla_\rho \sigma_{\mu\nu}=-\frac{2}{3}\theta \sigma_{\mu\nu}-\sigma_{\mu\rho}\sigma^{\rho}_{\nu}-\omega_{\mu\rho}\omega^{\rho}_{\nu}+\frac{1}{3}h_{\mu\nu}(\sigma_{\rho\sigma}\sigma^{\rho\sigma}-\omega_{\rho\sigma}\omega^{\rho\sigma})+C_{\rho\nu\mu\sigma}U^\rho U^\sigma +\frac{1}{2}\tilde R_{\mu\nu},
\end{equation}
where $\tilde R_{\mu\nu}$ is the spatial, trace-free part of $R_{\mu\nu}$, i.e,
\begin{equation}
    \tilde R_{\mu\nu}=h_{\mu\rho}h_{\nu\sigma}R^{\rho\sigma}-\frac{1}{3}h_{\mu\nu}h_{\rho\sigma}R^{\rho\sigma}.
\end{equation}
Finally, the antisymmetric part of eqn.\eqref{geodesic deviaiton equation}
\begin{equation}
    U^{\rho}\nabla_{\rho}\omega_{\mu\nu}=-\frac{2}{3}\theta\omega_{\mu\nu}-2\sigma^\rho_{[\nu}\omega_{\mu]\rho}.
\end{equation}
Eqn.\eqref{Raychaudhuri's equation} is known as Raychaudhuri's equation \cite{1955PhRv...98.1123R,Raychaudhuri1957SingularSI} and is the key equation used in the proof of the singularity theorems. We turn our attention, now, to investigate the positivity property of the last term on the right-hand side of this equation. Using Einstein's equations, we may write this term as
\begin{equation}\label{einsten equation ricci tnesor term}
    R_{\mu\nu}U^\mu U^\nu=8\pi\left[T_{\mu\nu}-\frac{1}{2}Tg_{\mu\nu}\right]U^\mu U^\nu=8\pi\left[T_{\mu\nu}U^{\mu}U^{\nu}+\frac{1}{2}T\right]
\end{equation}
Now, $T_{\mu\nu}U^{\mu}U^{\nu}$ physically represents the energy density of matter as measured by an observer whose 4-velocity is $U^\mu$. It is generically believed that for all physically reasonable classical matter this energy density is nonnegative, i.e,
\begin{equation}
    T_{\mu\nu}U^{\mu}U^{\nu}\geq 0
\end{equation}
for all timelike $U^\mu$. This assumption is known as the weak energy condition. However, it also seems physically reasonable that the stresses of matter will not become so large and negative as to make the right-hand side of eqn.\eqref{einsten equation ricci tnesor term} negative. This assumption that
\begin{equation}
    T_{\mu\nu}\geq -\frac{T}{2}
\end{equation}
for all unit timelike $U^\mu$ is known as the strong energy condition \cite{HawkingEllis1973,Wald1984,Curiel_2017}. 
\\
\\
Let us now return to Raychaudhuri's equation \eqref{Raychaudhuri's equation}. As discussed above, if Einstein's equation holds and the strong energy condition is satisfied by $T_{\mu\nu}$, then the last term of the right-hand side of eqn.\eqref{Raychaudhuri's equation} will be negative. This can be interpreted, physically, as a manifestation of the attractiveness of gravity. If the congruence is hypersurface orthogonal, we have $\omega_{\mu\nu}=0$, so the third term vanishes. The second term, $-\sigma_{\mu\nu}\sigma^{\mu\nu}$, is manifestly non-positive. Thus, under these assumptions, we find
\begin{equation}
    \frac{d\theta}{d\tau}+\frac{1}{3}\theta^2\leq0,
\end{equation}
which implies
\begin{equation}
    \frac{d}{d\tau}(\theta^{-1})\geq \frac{1}{3},
\end{equation}
and hence
\begin{equation}
    \theta^{-1}(\tau)\geq \theta_0^{-1}+\frac{1}{3}\tau,
\end{equation}
where $\theta_0$ is the initial value of $\theta$. Suppose, now, that $\theta_0$ is negative, i.e, the congruence is initially converging. Then eqn.\eqref{Raychaudhuri's equation} implies that $\theta^{-1}$ must pass through zero, i.e, $\theta\rightarrow-\infty$, within a proper time $\tau\leq 3/|\theta_0|$. Thus, we have proven the following lemma.
\begin{lemma}\label{Focusing within proper time}
    Let $U^\mu$ be the tangent field of a hypersurface orthogonal timelike geodesic congruence. Suppose $R_{\mu\nu}U^\mu U^\nu \geq0$, as will be the case if Einstein's equations hold in the spacetime and the strong energy condition is satisfied by the matter. If the expansion $\theta$ takes negative value $\theta_0$ at any point on the geodesic congruence, then $\theta$ goes to $-\infty$ along that geodesic within proper time $\tau\leq 3/|\theta_0|$.
\end{lemma}
In general, the singularity in $\theta$ implied by Lemma \ref{Focusing within proper time}, of course, represents merely a singularity in the congruence, not a singularity in the structure of spacetime.
\\
\\
We turn our attention, now, to null geodesic congruences. Again, we parameterize the geodesics by an affine parameter $\lambda$, but, unlike the timelike case, we now have no natural way of normalizing the tangent field $k^\mu$ and thereby adjusting the scale of $\lambda$ on different geodesics. In the timelike case, we restricted consideration to deviation vectors $V^\mu$ orthogonal to $U^\mu$. There actually were two independent reasons for doing so. 
\begin{itemize}
\item We have $U^\mu\nabla_\mu(U_\nu V^\nu)=U^\mu V_\nu\nabla_\mu V^\nu=U_\nu\mathcal{L}_U V^\nu+V^\mu U_\nu \nabla_\mu U^\nu=0$ provided that $U_\nu U^\nu$. Thus, $U_\mu V^\mu$ is constant along each geodesic, and the behavior of the “nonorthogonal” part of $V^\mu$ is uninteresting. 
\item Deviation vectors which differ only by a multiple of $U^\mu$ represent a displacement to the same nearby geodesic. Orthogonality fixes a natural “gauge condition” on $V^\mu$.
\end{itemize}
In the case of a null geodesic congruence, the above reasons for restricting the choice of the deviation vector still apply, but now they lead to two independent restrictions. First, for any deviation vector $V^\mu$, we may again have $k^\mu\nabla_\mu(k_\nu V^\nu)=k_\nu\mathcal{L}_kV^\nu+V^\mu k_\mu\nabla_\mu k^\mu=0$, so $k_\mu V^\mu$ doesn't vary along each geodesic. This implies that an arbitrary deviation vector $V^\mu$ may be written as the sum of a vector not orthogonal to $k^\mu$ which is parallelly propagated along the geodesic, plus a vector perpendicular to $k^\mu$. Thus, the behavior of the “nonorthogonal” part of $V^\mu$ again is uninteresting, and we may restrict consideration to deviation vectors satisfying $V_\mu k^\mu=0$. Second, deviation vectors which differ only by a multiple of $k^\mu$ again represent a displacement to the same nearby geodesic. Thus, the physically interesting quantity is really the equivalence class of deviation vectors, where two deviation vectors are considered equivalent if their difference is a multiple of $k^\mu$. Since $k^\mu$ is null and thus orthogonal to itself, this second restriction is independent of the first restriction, and it reduces the physically interesting class of deviation vectors to a two-dimensional subspace, as we shall now describe with more precision.
\\
\\
Let $V_p$ denote the tangent space at a point $p\in \mathcal{M}$. The tangent vectors in $V_p$ which are orthogonal to $k^a$ form a three-dimensional subspace, which we shall denote by $\tilde{V}_p$. We define $\tilde V_p$ to be the vector space of equivalence classes of vectors in $\tilde V_p$, where $x^a,y^a\in \tilde V_p$ are defined to be equivalent if there exists a number $c\in\mathbb{R}$ such that $x^\mu-y^\mu=ck^\mu$. Then $\hat{V}_p$ is a two-dimensional vector space, although it can be identified as a subspace of $V_p$ in a natural way. As explained above, the deviation vectors of interest live in the vector space $\hat{V}_p$.
\\
\\
A tangent vector $t^\mu\in V_p$ does not naturally give rise to a tangent vector space in $\tilde V_p$, because there is no natural way of decomposing it as a sum of a vector in $\tilde V_p$ and a vector not lying in $\tilde{V}_p$. However, if $t^\mu\in\tilde{V}_p$ (i.e., if $t^\mu k_\mu=0$), then $t^\mu$ naturally gives rise to a vector $\hat t^\mu$ in $\hat V_p$ taking its equivalence class. On the other hand, a dual vector $\mu_\nu\in V_p^*$ naturally gives rise to a dual vector $\tilde \mu_a \in (\tilde V_p)^*$ by restricting its action to vectors in $\tilde V_p$. However, $\tilde \mu_\nu$ naturally gives rise to a dual vector $\hat \mu _\nu \in (\hat V_p)^*$ if and only if $\tilde \mu_\nu k^\nu =\mu_\nu k^\nu =0$. More generally, a tensor $T^{\mu_1\dots \mu_n}_{\nu_1\dots \nu_n}$ over $V_p$ naturally gives rise to a tensor  $\hat T^{\mu_1\dots \mu_n}_{\nu_1\dots \nu_n}$ over $\hat V_p$ if and only if the result of contracting any one of its indices with $k_\mu$ or $k^\mu$ and contracting its remaining indices with vectors or dual vectors which give rise to elements of $\hat V_p$ or $(\hat V_p)^*$ is always zero. For tensors satisfying this property, the operation of taking outer products commutes with the "projection" into tensors over $\hat V_p$. Furthermore, a tensor satisfying the stronger property that contraction over any one of its indices into $k^\mu$ or $k_\mu$ gives zero, then projection into $\hat V_p$ commutes with contraction.
\\
\\
The spacetime metric $g_{\mu\nu}$ satisfying the above property and this gives rise to a tensor over $\hat V_p$, which we shall denote as $\hat h_{\mu\nu}$. It is not difficult to see that $\hat h_{\mu\nu}$ is a positive definite (i.e, signature $+ \space +$) metric on $\hat V_p$. (On the other hand, the tensor $\tilde h_{\mu\nu}$ over $\hat V_p$ obtained by restricting the action of $g_{\mu\nu}$ to vectors in $\hat V_p$ is degenerate, since $\tilde h_{\mu\nu} k^\nu=0$, so $\tilde h_{\mu\nu}$ is not a metric on $\tilde V_p$). Note that the inverse metric $\hat h^{\mu\nu}$ is just the “hatted” tensor \cite{Sachs:1962wk,Sachs1961_VI} associated with $g^{\mu\nu}$.
\\
\\
Consider, now, a congruence of null geodesics with tangent field $k^\mu$. The tensor field
\begin{equation}
    B_{\mu\nu}=\nabla_\nu k_\mu,
\end{equation}
also satisfies the above property and thus gives rise to a “hatted” tensor field $\hat B_{\mu\nu}$. We can decompose $\hat B_{\mu\nu}$ as
\begin{equation}\label{Hatted decomposition}
    \hat B_{\mu\nu}=\frac{1}{2}\theta \hat h_{\mu\nu}+\hat \sigma_{\mu\nu}+\hat \omega_{\mu\nu},
\end{equation}
with
\begin{equation}
    \theta =\hat P^{\mu\nu}\hat B_{\mu\nu},
\end{equation}
\begin{equation}
    \hat \sigma_{\mu\nu}=\hat B_{(\mu\nu)}-\frac{1}{2}\theta \hat P_{\mu\nu},
\end{equation}
\begin{equation}
    \hat \omega_{\mu\nu}=\hat B_{[\mu\nu]},
\end{equation}
so that $\theta,\hat\sigma_{\mu\nu},\hat \omega_{\mu\nu}$ again have the physical interpretation as, respectively, the expansion, shear, and twist of the congruence. The change in the numerical factor in the term $\frac{1}{2}\theta \hat P_{\mu\nu}$ (as compared with $\frac{1}{3}\theta \hat P_{\mu\nu}$ in the timelike case) arises simply because the relevant vector space is now two-dimensional rather than three-dimensional. The same derivation that led to eq.\eqref{geodesic deviaiton equation} in the timelike case now yields
\begin{equation}
    k^\rho\nabla_\rho B_{\mu\nu}+B^{\rho}_\nu B_{\mu\rho}=R_{\rho\nu\mu}^{\sigma}k_\sigma k^\rho .
\end{equation}
“Hatting” this equation, we obtain
\begin{equation}
    k^\rho \nabla_\rho \hat B_{\mu\nu}+\hat B^\rho_\nu \hat B_{\mu\rho}=\widehat{ R_{\sigma\rho \nu \mu}k^{\rho}k^{\sigma}}.
\end{equation}
Finally, substituting our decomposition equation\eqref{Hatted decomposition} for $\hat B_{\mu\nu}$ and taking the trace, symmetric trace-free, and antisymmetric parts of the equation, we obtain, respectively,
\begin{equation}\label{Null Raychaudhri}
    \frac{d\theta}{d\lambda}=-\frac{1}{2}\theta^2-\hat \sigma_{\mu\nu}\hat \sigma^{\mu\nu} +\hat \omega_{\mu\nu}\hat \omega^{\mu\nu}-R_{\rho \sigma}k^\rho k^\sigma,
\end{equation}
\begin{equation}
    k^\rho\nabla_\rho \hat \sigma_{\mu\nu}=-\theta \hat \sigma_{\mu\nu}+\widehat{C_{\rho \nu\mu\sigma}k^\rho k^\sigma},
\end{equation}
\begin{equation}
    k^\rho \nabla_\rho \hat \omega_{\mu\nu}=-\theta\hat \omega_{\mu\nu}.
\end{equation}
The nature of the equation \eqref{Null Raychaudhri} for the expansion of null geodesics is very similar to that of Raychaudhuri's equation \eqref{Raychaudhuri's equation}. Apart from the change in the factor $1/3$ to $1/2$, the only significant change is that, using Einstein's equation, we obtain
\begin{equation}
    R_{\mu\nu}k^\mu k^\nu=8\pi T_{\mu\nu}k^\mu k^\nu .
\end{equation}
Thus, all that is needed to ensure the last term of eqn.\eqref{Null Raychaudhri} is nonpositive is that for all null $k^\mu$, 
\begin{equation}\label{weak energy}
    T_{\mu\nu}k^\mu k^\nu\geq 0.
\end{equation}
If the strong energy condition holds, then for all timelike $U^\mu$ we have $T_{\mu\nu}U^\mu U^\nu-\frac{1}{2}TU^\mu U_\mu \geq 0$, and hence, for null-like $k^\mu$, this is just equivalent to eqn.\eqref{weak energy}. By the same argument that led to Lemma \ref{Focusing within proper time}, we obtain the following result
\begin{lemma}
    Let $k^\mu$ be the tangent field of a hypersurface-orthogonal null geodesic congruence. Suppose $R_{\mu\nu}k^\mu k^\nu \geq 0$, as will be the case if Einstein's equation holds in the spacetime and the weak or strong energy condition is satisfied by that matter. If the expansion $\theta$ takes the negative value $\theta_0$ at any point on a geodesic in the congruence, then $\theta $ goes to $-\infty$ along that geodesic within affine length $\lambda \leq 2/|\theta_0|$.
\end{lemma}
Now, let $\mathcal{M}$ be any manifold on which a connection is defined and let $\gamma$ be a geodesic with tangent $U^\mu$. A solution $V^\mu$, of the geodesic deviation equation
\begin{equation}\label{geodesic deviation}
    U^\mu \nabla_\mu (U^\nu \nabla_\nu V^\rho)=-R_{\mu\nu \sigma}^{\rho}V^\nu U^\mu U^\sigma,
\end{equation}
is called a Jacobi field on $\gamma$. A pair of points $p,q\in \gamma$ are said to be conjugate if there exists a Jacobi field $V^\mu$ which is not identically zero but vanishes at both $p$ and $q$ \cite{beem2017global,oneill1983semiriemannian,lee2012introduction,lee2018introduction}. 
\\
\\
Let $\gamma$ be a timelike geodesic with tangent $U^\mu$ and let $p\in \gamma$. Consider the congruence of all timelike geodesics passing through $p$. (This congruence, of course, is singular at $p$ itself.) Then every Jacobi field which vanishes at $p$ is a deviation vector for this congruence. We shall show that a point $q\in \gamma$ lying to the future of $p$ is conjugate to $p$ if and only if, the expansion, $\theta$, of this congruence approaches $-\infty$ at $q$. For this purpose, it is convenient to introduce an orthonormal basis of spatial vectors $e_1^a,e_2^a,e_3^a$ orthogonal to $V^\mu$ and parallelly propagated along $\gamma$. Since the deviation vector satisfies a linear ordinary differential equation (eqn.\eqref{geodesic deviation}), the value of $V^\mu$ at time $\tau$ must depend linearly on the initial data $V^\mu(\tau)$ and $dV^\mu/d\tau \space(0)$ at $p$. Since, by construction, $V^\mu (0)=0$ for this congruence, we must have
\begin{equation}\label{form ivp}
    V^\mu (\tau)= A^\mu_\nu(\tau) \frac{dV^\nu}{d\tau}(0).
\end{equation}
Substituting this into eqn.\eqref{geodesic deviation}, we see that the matrix $A^\mu_\nu (\tau)$ satisfies the equation
\begin{equation}
    \frac{d^2 A^\mu_\nu}{d\tau^2}=-R_{\rho\sigma\alpha}^\mu U^\rho U^\alpha A^\sigma_\nu.
\end{equation}
Clearly, we also have $A^\mu_\nu (0)=0$ and $d A^\mu_\nu/d\tau (0)=\delta^\mu_\nu$. Now, $q$ will be conjugate to $p$ if and only if there exists nontrivial initial data [i.e, $dV^\mu/d\tau \not =0$] for which $V^\mu =0$ at $q$. By eqn. \eqref{form ivp}, this occurs if and only if $\text{det} A^\mu_\nu=0$ at $q$, so $\text{det} A^\mu_\nu=0$ is the necessary and sufficient condition for a conjugate point to $p$. Note that between conjugate points we have $\text{det}A^\mu_\nu\not=0$, so that the inverse of $A^\mu_\nu$ exists. 
\\
\\
Clearly, the matrix $A^\mu_\nu$ must be related ti the tensor $B_{\mu\nu}=\nabla_\nu U_\mu$ of the congruence. To find this relation, we note that
\begin{equation}
    \frac{d V^\mu}{d\tau}=B^\mu_\rho V^\rho.
\end{equation}
However, using eqn.\eqref{form ivp}, we find
\begin{equation}
    \frac{dV^\mu}{d\tau}=\frac{dA^\mu_\nu}{d\tau}\frac{dV^\nu}{d\tau}(0),
\end{equation}
so we obtain
\begin{equation}
    \frac{dA^\mu_\nu}{d\tau}\frac{dV^\nu}{d\tau}(0)=B^\mu_\rho A^\rho_\nu \frac{dV^\nu}{d\tau}(0).
\end{equation}
Thus, in matrix notation, we find
\begin{equation}
    \frac{dA}{d\tau}= B A,
\end{equation}
i.e,
\begin{equation}
    B=\frac{d A}{d\tau} A^{-1}.
\end{equation}
Consequently, we obtain
\begin{equation}
    \theta=\text{Tr} B=\text{Tr}\left[\frac{dA}{d\tau}A^{-1}\right],
\end{equation}
where $\text{Tr}$ denotes the trace of a matrix. However, it follows from the formula for the inverse of any matrix that for any nonsingular matrix $A$, we have
\begin{equation}
    \text{Tr}\left[\frac{d A}{d\tau}A^{-1}\right]=\frac{1}{\text{det}A}\frac{d}{d\tau}(\text{det}A),
\end{equation}
so that
\begin{equation}\label{theta det A}
    \theta=\frac{d}{d\tau}(\text{ln}|\det A|).
\end{equation}
Since $A$ satisfies an ordinary differential equation, $d(\text{det} A)/d\tau$ cannot be infinite anywhere along $\gamma$. Therefore, if $\theta\rightarrow -\infty$ at $q$, it follows from eqn.\eqref{theta det A} that $\det A\rightarrow 0$ at $q$. Conversely, if $\det A\rightarrow 0$ at $q$, it follows that $\theta\rightarrow -\infty$ at $q$. Thus, as desired, we have proven that a necessary and sufficient condition for $q$ to be conjugate to $p$ is that for the congruence of all timelike geodesics emanating from $p$, we have $\theta\rightarrow -\infty$ at $q$. 
We may use Lemma \ref{Focusing within proper time} to establish the following result on the existence of conjugate points.
\begin{prop}
    Let ($\mathcal{M}$,$g_{\mu\nu}$) be a spacetime satisfying $R_{\mu\nu}U^\mu U^\nu\geq0$ for all timelike $U^\mu$. Let $\gamma$ be a timelike geodesic and let $p\in \gamma$. Suppose the convergence of the congruence of timelike geodesics emanating into the future from $p$ attains the negative value $\theta_0$ at $r\in\gamma$. Then within proper time $\tau\leq 3/|\theta_0|$ from $r$ along $\gamma$ there exists a point $q$ conjugate to $p$, assuming, of course, that $\gamma$ extends that far.
\end{prop}
Next, we turn our attention to the relation between conjugate points and the extremal length properties of timelike geodesics. Let $p,q\in\mathcal{M}$ and consider a smooth one-parameter family of smooth timelike curves $\lambda_a(t)$ from $p$ to $q$, where the curve parameter $t$ is chosen so that for all $\alpha$ we have $\lambda_\alpha(a)=p,\lambda_\alpha(b)=q$. We denote the tangent vectors by $(\partial/\partial t)^\mu$, by $T^\mu$ and the deviation vectors, $(\partial/\partial\alpha)^\mu$, by $X^\mu$. Then $X^\mu$ vanishes at both $p$ and $q$ and, as usual, we have $\mathcal{L}_T X^\mu=T^\nu\nabla_\nu X^\mu -X^\nu\nabla_\nu T^\mu=0$ everywhere. The length of each curve is given by \cite{oneill1983semiriemannian,Wald1984,HawkingEllis1973}
\begin{equation}
    \tau(\alpha)=\int_a^b f(\alpha,t)dt,
\end{equation}
where $f=(-T^\mu T_\mu)^{1/2}$. We have 
\begin{equation}
\begin{split}
    \frac{d\tau}{d\alpha}=\int_a^b \frac{\partial f}{\partial \alpha} dt\\
    =\int_a^b X^\mu\nabla_\mu (-T^\nu T_\nu)^{1/2}dt\\
    =-\int_a^b \frac{1}{f}T_\nu X^\mu \nabla_\mu T^\nu dt \\
    = -\int_a^b \frac{1}{f}T_\nu T^\mu \nabla_\mu X^\nu dt\\=-\int_a^b T^\mu\nabla_\mu\left[\frac{1}{f}T_\nu X^\nu\right]dt +\int_a^b X^\nu T^\mu\nabla_\mu (T_\nu /f) dt\\=\int_a^b X^\nu T^\mu \nabla_\mu(T_\nu /f) dt,
\end{split}
\end{equation}
since $T^\mu\nabla_\mu (f^{-1}T_\nu X^\nu)=\partial (f^{-1}T_\nu X^\nu)/\partial t$ and $X^b$ vanishes at the endpoints. Thus, setting $\alpha=0$ we see that $d\tau/d\alpha=0$ for arbitrary $X^\nu$ if and only if $T^\mu\nabla_\mu (T_\nu/f)=0$ at $\alpha=0$, which is just the geodesic equation expressed in an arbitrary parameterization.
\\
\\
The second variation of arc length is
\begin{equation}
    \frac{d\tau^2}{d\alpha^2}=\int_a^b X^\rho \nabla_{\rho} [X^\nu T^\mu\nabla_\mu(T_\nu/f)]dt.
\end{equation}
Evaluating this expression at $\alpha=0$, assuming that $\lambda_0$ is a geodesic, we find
\begin{equation}\label{proper time second deivative 1}
    \begin{split}
        \frac{d^2\tau}{d\alpha^2}\bigg|_{\alpha=0}=\int_a^b X^\nu (X^\rho \nabla_\rho T^\mu)\nabla_\mu (T_\nu/f) dt +\int_a^b X^\nu T^\mu X^\rho \nabla_\rho \nabla_\mu (T_\nu /f) dt \\
        =\int_a^b X^\nu (T^\rho \nabla_\rho X^\mu)\nabla_\mu (T_\nu /f) dt +\int_a^b X^\nu T^\mu X^\rho \nabla_\mu \nabla_\rho (T_\nu f) dt+\int_a^b \frac{X^\nu T^\mu X^\rho R_{\rho \mu \nu}^\sigma T_\sigma}{f}dt \\
        =\int_a^b X^\nu T^\rho\nabla_\rho [X^\mu\nabla_\mu (T_\nu/f)] dt+\int_a^b \frac{X^\nu T^\mu X^\rho R_{\rho\mu\nu}^{\sigma} T_{\sigma}}{f}f dt.
    \end{split}
\end{equation}
The term in the square brackets can be reexpressed as
\begin{equation}
\begin{split}
    X^\mu\nabla_\mu (T_\nu/f)=\frac{1}{f}X^\mu\nabla_\mu T_\nu -\frac{1}{f^2}T_\nu X^\mu\nabla_\mu f\\
    =\frac{1}{f}T^\mu\nabla_\mu X_\nu +\frac{1}{f^2}T_\nu T^\mu \nabla_\mu\left[\frac{1}{f}T_\rho X^\rho\right].
\end{split}
\end{equation}
Eqn.\eqref{proper time second deivative 1} can be put in a more easily recognizable form by choosing the curve parametrization so that $f=1$ along the geodesic $\lambda_0$ and also choosing the deviation vector $X^\mu$ to be orthogonal to $T^\mu$ along $\lambda_0$. With these choices, we find
\begin{equation}
    \begin{split}
         \frac{d^2\tau}{d\alpha^2}\bigg|_{\alpha=0}=\int_a^b X^\nu\{T^\rho\nabla_\rho(T^\mu\nabla_\mu X_\nu)+R_{\rho\mu\nu}^{\sigma}T^\mu T_{\sigma}X^\rho \} dt\\
         =\int_a^b X^\nu (\mathcal{O}X)_\nu dt,
    \end{split}
\end{equation}
where $\mathcal{O}$ is the operator appearing in the geodesic deviation equation. Using this formula, we may establish the following theorem.
\begin{thm}\label{maximize conjuagte points}
    Let $\gamma$ be a smooth timelike curve connecting two points $p,q\in \mathcal{M}$. Then the necessary and sufficient condition that $\gamma$ locally maximizes the proper time between $p$ and $q$ over smooth one-parameter variations is that $\gamma$ be a geodesic with no point conjugate to $p$ between $p$ and $q$.
\end{thm}
A notation similar to conjugacy of two similar points along a timelike geodesic can be defined for a point and a smooth (or, at least,$C^2$) spacelike hypersurface $\Sigma$. First, we define the extrinsic curvature $K_{\mu\nu}$ of $\Sigma$. Let $U^\mu$ be the unit tangent field of the congruence of timelike geodesics orthogonal to $\Sigma$. We define $K_{\mu\nu}$ by
\begin{equation}
    K_{\mu\nu}=\nabla_\mu U_\nu =B_{\nu\mu},
\end{equation}
where the evaluation of these tensors on $\Sigma$  is understood. Note that $K_{\mu\nu}$ is purely spatial, i.e, $K_{\mu\nu}U^\mu=K_{\mu\nu}U^\nu=0$. Since this congruence is manifestly hypersurface orthogonal, we have $\omega_{\mu\nu}=0$ so $K_{\mu\nu}$ is symmetric, $K_{\mu\nu}=K_{\nu\mu}$. Hence, we can rexpress $K_{\mu\nu}$ as 
\begin{equation}
\begin{split}
    K_{\mu\nu}=\frac{1}{2}\mathcal{L}_{U}g_{\mu\nu}\\
    =\frac{1}{2}\mathcal{L}_u (h_{\mu\nu}-U_\mu U_\nu)\\
    =\frac{1}{2}\mathcal{L}_U h_{\mu\nu}.
\end{split}
\end{equation}
Now, $h_{\mu\nu}$ is the spatial metric induced on the hypersurface of constant proper time from $\Sigma$ along the geodesic congruence orthogonal to $\Sigma$. Thus, $K_{\mu\nu}$ measures the rate of change of this spatial metric as one moves along the congruence; i.e, it measures the “bending” of $\Sigma$ in the spacetime ($\mathcal{M},g_{\mu\nu}$). In Gaussian normal coordinates, we have
\begin{equation}
    K_{\mu\nu}=\frac{1}{2}\frac{\partial h_{\mu\nu}}{\partial t},
\end{equation}
which reinforces our interpretation of $K_{\mu\nu}$ representing the “time derivative” of $h_{\mu\nu}$. The trace of the extrinsic curvature is often denoted by $K$,
\begin{equation}
    K=K^{\mu}_{\mu}=h^{\mu\nu}K_{\mu\nu}.
\end{equation}
Thus, we have
\begin{equation}
    K=\theta,
\end{equation}
where $\theta$ is the expansion of the geodesic congruence orthogonal to $\Sigma$. A point $p$ on a geodesic $\gamma$ is said to be conjugate to $\Sigma$ along $\gamma$ is there exists an orthogonal deviation vector $V^\mu$ if the congruence which is nonzero on $\Sigma$ but vanishes at $p$. By the same argument as given above for conjugate pair of points, $p$ will be conjugate to $\Sigma$ if and only if the expansion, $\theta$, of the congruence of geodesics orthogonal to $\Sigma$ approaches $-\infty$ at $p$. Since this congruence is manifestly hypersurface orthogonal, Lemma \ref{Focusing within proper time} immediately yields the following result
\begin{prop}\label{focsuing extrinsic curvature}
    Let $(\mathcal{M},g_{\mu\nu})$ be a spacetime satisfying $R_{\mu\nu}U^\mu U^\nu\geq0$ for all timelike $U^\mu$. Let $\Sigma$ be a spacelike hypersurface with $K=\theta<0$ at a point $q\in \Sigma$. Then within proper time $\tau \leq 3/|K|$ there exists a point $p$ conjugate to $\Sigma$ along the geodesic $\gamma$ orthogonal to $\Sigma$ and passing through $q$, assuming that $\gamma$ can be extended that far.
\end{prop}
Furthermore, by arguments similar to those used to prove Theorem \ref{maximize conjuagte points}, the following theorem can be proven
\begin{thm}\label{maximise proper time conjugate surface}
    Let $\gamma$ be a smooth timelike curve connecting a point $p\in\mathcal{M}$ to a point $q$ on a smooth spacelike hypersurface $\Sigma$. Then the necessary and sufficient condition that $\gamma$ locally maximizes the proper time between $p$ and $\Sigma$ over smooth one-parameter variations is that $\gamma$ be a geodesic orthogonal to $\Sigma$ with no conjugate point to $\Sigma$ between $\Sigma$ and $p$.
\end{thm}
We turn our attention, now, to conjugate points of null geodesics. It follows directly from the geodesic deviation equation that for any Jacobi field $V^\mu$ on a null geodesic $\mu$ with tangent $k^\mu$, we have
\begin{equation}
    k^\rho\nabla_\rho[k^\nu\nabla_\nu(k^\mu V_\mu)]=0,
\end{equation}
which immediately implies that $V^\mu$ cannot vanish at two points $p,q\in\mu$ unless $k^\mu V_]mu=0$ everywhere along $\mu$. Furthermore, if $V^\mu$ is a Jacobi field, then so is $V^\mu+(a+b\lambda) k^\mu$, where $a$ and $b$ are constants, so $p$ and $q$ will be conjugate if and only if there exists a Jacobi field $V^\mu$ which differs from zero by only a multiple of $k^\mu$ at both $p$ and $q$. Thus, along a null geodesic $\mu$, the points $p,q\in \mu$ will be conjugate if and only if $\hat V^\mu$ in $\hat V$ satisfies the geodesic deviation equation and vanishes at $p$ and $q$. All such $\hat V^\mu$ which vanish at $p$ arise as deviation vectors of any null geodesic congruence containing the two-dimensional family of null geodesics emerging from $p$. By the same arguments as given above in the timelike case, $q$ will be conjugate to $p$ if and only if the expansion, $\theta$, of such a null geodesic congruence approaches $-\infty$ at $q$. Thus, Lemma \ref{Focusing within proper time} implies the following result
\begin{prop}
    Let $(\mathcal{M},g_{\mu\nu})$ be a spacetime satisfying $R_{\mu\nu}k^\mu k^\nu\geq 0$ for all null $k^\mu$. Let $\mu$ be a null geodesic and let $p\in \mu$. Suppose the convergence $\theta$ of the null geodesics emanating from $p$ attains the negative value $\theta_0$ at $r\in\mu$. Then, within affine length $\lambda\leq 2/|\theta_0|$ from $r$, there exists a point $q$ conjugate to $p$ along $\mu$, assuming that $\mu$ extends that far.
\end{prop}
In analogy to Theorem \ref{maximize conjuagte points}, we obtain the following result, 
\begin{thm}
    Let $\mu$ be a smooth causal curve and let $p,q\in\mu$. Then there does not exist a smooth one-parameter family of causal curves $\lambda_a$ connecting $p$ and $q$ with $\lambda_0=\mu$ and $\lambda_\alpha$ timelike for all $\alpha>0$ (i.e,$\mu$ cannot be smoothly deformed to a timelike curve) if and only if $\mu$ is a null geodesic with no point conjugate to $p$ along $\mu$ between $p$ and $q$.
\end{thm}
For null geodesics, a notion of conjugacy can also be defined for a point and a two-dimensional spacelike surface (i.e. embedded submanifold) $S$. At each $q\in S$ there will exist precisely two future-directed null vectors $k_1^\mu,k_2^\mu$ which are orthogonal to $S$. If $S$ is orientable, we can make a continuous choice of $k_1^\mu$ and $k_2^\mu$ over $S$ and thereby define two families of null geodesics, which we may refer to as “outgoing” and “ingoing” families. We will refer to each of these families as a congruence even though each spans only a null hypersurface rather than an open region of spacetime. The expansion $\theta$, shear $\hat \sigma_{\mu\nu}$, and twist $\hat \omega_{\mu\nu}=0$ of these congruences are well defined, since all deviation vectors orthogonal to the geodesic tangents $k^\mu$ are included in the congruences. Let $\mu$ be a null geodesic in one of these congruences. A point $p\in\mu$ is said to be conjugate to $S$ along $\mu$ if along $\mu$ there exists a deviation vector $\hat V^\mu$ of the congruence which is nonzero on $S$ but vanishes at $p$. In analogy to Proposition \ref{focsuing extrinsic curvature}, we have
\begin{prop}\label{null conjugate}
    Let $(\mathcal{M},g_{\mu\nu})$ be a spacetime satisfying $R_{\mu\nu}k^\mu k^\nu\geq0$ for all null $k^\mu$. Let $S$ be a smooth two-dimensional spacelike submanifold such that the expansion, $\theta$, of, say, the “outgoing” null geodesics has the negative value $\theta_0$ at $q\in S$. Then, within affine parameter $\lambda\leq 2/|\theta_0|$, there exists a point $p$ conjugate to $S$ along the outgoing null geodesic $\mu$ passing through $q$.
\end{prop}
In analogy to Theorem \ref{maximise proper time conjugate surface}, we also have:
\begin{thm}
    Let $S$ be a smooth two-dimensional spacelike submanifold and let $\mu$ be a smooth causal curve from $S$ to $p$. Then the necessary and sufficient condition that $\mu$ cannot be smoothly deformed to a timelike curve connecting $S$ and $p$ is that $\mu$ is a null geodesic orthogonal to $S$ with no point conjugate to $S$ between $S$ and $p$.
\end{thm}
As a consequence of this theorem and results regarding the causal structure of spacetime \cite{HawkingEllis1973,penrose1972techniques,Minguzzi2008Causal,Minguzzi2019Lorentzian,Bernal:2004gm,BernalSanchez2003,Bernal_2006}, we have the following theorem
\begin{thm}\label{null no conjugate surface}
    Let ($\mathcal{M},g_{\mu\nu}$) be a globally hyperbolic spacetime and let $K$ be a compact orientable, two-dimensional spacelike submanifold $\mathcal{M}$. Then every $p\in\dot{I}(K)$, where $\dot{I}(K)$ denotes the boundary of the chronological future of $K$, lies on a future-directed null geodesic starting from $K$ which is orthogonal to $K$ and has no point conjugate to $K$  between $K$ and $p$.
\end{thm}
We now give the following theorems regarding the existence of maximum length curves, which can be proven independently of the above results, and from results regarding the causal structure of spacetime itself. 
\begin{thm}\label{maximum geodesic orthogonal no conjugate}
    Let ($\mathcal{M},g_{\mu\nu}$) be a strongly causal spacetime. Let $p\in \mathcal{M}$, let $\Sigma$ be an achronal, smooth spacelike hypersurface, and consider the length function $\tau$ defined on $C(p,q)$. A necessary condition for $\tau$ to attain its maximum value at $\gamma\in C(p,q)$ is that $\gamma$ be a geodesic with no point conjugate to $p$ between $p$ and $q$.
\end{thm}
\begin{thm}
    Let $(\mathcal{M},g_{\mu\nu})$ be a globally hyperbolic spacetime. Let $p,q\in\mathcal{M}$ with $q\in J^+(p)$. Then there exists a curve $\gamma\in C(p,q)$ for which $\tau$ attains its maximum value on $C(p.q)$.
\end{thm}
\begin{thm}\label{golbal hyperbolic surface maximum}
    Let $(\mathcal{M},g_{\mu\nu})$ be a globally hyperbolic spacetime. Let $p\in \mathcal{M}$ and let $\Sigma $ be a Cauchy surface. Then there exists a curve $\gamma\in C(\Sigma,p)$ for which $\tau$ attains its maximum value on $C(\Sigma,p)$.
\end{thm}
We have now developed the machinery required to prove some of the singularity theorems \cite{PhysRevLett.14.57,PhysRevLett.17.444,Hawking:1967ju,Hawking:1970zqf,Galloway:2005mf,Galloway2000MaxPrinciples,Beran2022Splitting,Senovilla_2022,Senovilla1998,Witten_2020}. 
\\
\\
The first theorem we shall prove can be interpreted as showing that if the universe is globally hyperbolic and at one instant is expanding everywhere at a rate bounded away from zero, then the universe must have begun in a singular state a finite time ago.
\begin{thm}
    Let ($\mathcal{M},g_{\mu\nu}$) be a globally hyperbolic spacetime with $R_{\mu\nu}U^\mu U^\nu\geq 0$ for all timelike $U^\mu$, which will be the case if Einstein's equation is satisfied with the strong energy condition holding for matter. Suppose there exists a smooth (or, at least $C^2$) spacelike Cauchy surface $\Sigma$ for which the trace of the extrinsic curvature (for the past-directed normal geodesic congruence) satisfies $K\leq C<0$ everywhere, where $C$ is a constant. Then no past-directed timelike curve from $\Sigma$ can have length greater than $3/|C|$. In particular, all past-directed timelike geodesics are incomplete.
\end{thm}
\begin{proof}
    Suppose there were a past-directed timelike curve $\lambda$ from $\Sigma$ with length greater than $3/|C|$. Let $p$ be a point on $\lambda$ lying beyond length $3/|C|$ from $\Sigma$. By Theorem \ref{golbal hyperbolic surface maximum}, there exists a maximum-length curve $\gamma$ form $p$ to $\Sigma$, which clearly also must have a length greater than $3/|C|$. By Theorem \ref{maximum geodesic orthogonal no conjugate}, $\gamma$ must be a geodesic with no conjugate point between $\Sigma$ and $p$. However, this contradicts Proposition \ref{focsuing extrinsic curvature} which states that $\gamma$ must have a conjugate point between $\Sigma$ and $p$. Therefore, the original curve $\lambda$ cannot exist.
\end{proof}
The next theorem proves null geodesic incompleteness in a context relevant to gravitational collapse. A compact, two-dimensional, smooth spacelike submanifold, $T$, having the property that the expansion, $\theta$, of both sets (i.e., “ingoing” and “outgoing”) of future-directed null geodesics orthogonal to $T$ is everywhere negative is called a trapped surface \cite{PhysRevLett.14.57,Hawking:1970zqf,Hawking:1967ju,penrose1972techniques,joshi1993global,Hayward:1993hg,Hayward:1997jp,Hayward:2004md}. Through Proposition \ref{null conjugate} and Theorem \ref{null no conjugate surface} and additional results regarding the causal structure of spacetime, we have
\begin{thm}
   Let ($\mathcal{M},g_{\mu\nu}$) be a connected, globally hyperbolic spacetime with a noncompact Cauchy surface $\Sigma$. Suppose $R_{\mu\nu}k^\mu k^\nu\geq 0$ for all null $k^\mu$, as will be the case if ($\mathcal{M},g_{\mu\nu}$) is a solution of Einstein's equation with matter satisfying the strong energy condition. Suppose further that $\mathcal{M}$ contains a trapped surface $T$. Let $\theta_0<0$ denote the maximum value of $\theta$ for both sets of orthogonal geodesics in $T$. Then at least one inextendible future-directed orthogonal null geodesic from $T$ has affine length no greater than $2/|\theta_0|$.
\end{thm}
\section{Non-Commutativity in Geodesic Congruences}\label{Section 3}
In the standard consideration of geodesic congruences as highlighted in Section\ref{Section 2}, the tangent and separation vectors are assumed to commute,
\begin{equation}
    [U,V]^{\mu}=\mathcal{L}_U V^\mu=U^{\lambda}\nabla_{\lambda}V^{\mu}-V^{\lambda}\nabla_{\lambda}U^{\mu}=0.
\end{equation}
One may naturally ask what happens if this assumption is relaxed. Such a possibility is motivated by the considerable interest that spacetime non-commutativity has received over the past several decades, particularly in non-commutative quantum field theory and quantum gravity. In most existing approaches, however, non-commutativity is introduced by replacing the underlying smooth manifold with a non-commutative algebra of coordinates, as in Connes non-commutative geometry and related formulations. Consequently, the discussion is typically algebraic rather than geometric, and the conventional notion of a smooth spacetime manifold is no longer fundamental.
\\
\\
Nevertheless, this need not be the only realization of spacetime non-commutativity. More generally, one may consider a spacetime that remains a smooth differentiable manifold while allowing the commutation properties of neighboring geodesics to be modified. This follows from the Frobenius theorem \cite{frobenius1877ueber,warner1983foundations}. Let us define the following notions \cite{lee2012introduction,lee2018introduction,doCarmo1992Riemannian,spivak1999comprehensive,warner1983foundations} before introducing the same,
\begin{defn}
    Let $c$ be an integer, $1\leq c\leq d$. A c-dimensional distribution $\mathcal{D}$ on a $d$-dimensional manifold $\mathcal{M}$ is a choice of a c-dimensional subspace $\mathcal{D}(m)$ of $\mathcal{M}$. $\mathcal{D}$ is smooth if for each $m$ in $M$ there is a neighborhood $U$ of $m$ and there are $c$ vector fields $X_1,\dots X_c$ of class $C^{\infty}$ on $U$ which span $\mathcal{D}$ at each point of $U$. A vector field $X$ on $M$ is said to belong to (or lie in) the distribution $\mathcal{D}$ ($X\in\mathcal{D}$) id $X_m\in\mathcal{D}$ for each $m\in M$. A smooth distribution $\mathcal{D}$ is called involutive if $[X,Y]\in \mathcal{D}$ whenever $X$ and $Y$ are smooth vector fields lying in $\mathcal{D}$.
\end{defn}
\begin{defn}
    A submanifold $(N,\psi)$ of $M$ is an integrable manifold of a distribution $\mathcal{D}$ is
    \begin{equation}
        d\psi(T_nN)=\mathcal{D}(T_n\psi),
    \end{equation}
    for every $n\in N$. Here, $\psi:N\rightarrow M$ is the embedding of the submanifold $N$ into the ambient manifold $M$ and $d\psi$ is the pushforward map, $d\psi:T_nN\rightarrow T_{\psi(n)}M$.
\end{defn}
We can now introduce Frobenius theorem as
\begin{thm}
    Let $m\in M$ then, $\mathcal{D}$ is a involutive, $C^{\infty}$ distribution on $M^d$ if and only if there exists an integral distribution $\mathcal{D}$ passing through $m$. 
\end{thm}
In our case, the relevant distribution is $\mathcal{D}=\text{span}(U,V)$. Hence, even if $\mathcal{D}$ is an integral distribution, $[U,V]^{\mu}\not=0$ in general rather,
\begin{equation}
    \mathcal{L}_U V^\mu =[U,V]^{\mu}=U^{\lambda}\nabla_{\lambda}V^{\mu}-V^{\lambda}\nabla_{\lambda}U^{\mu}=\kappa^{\rm NC} U^{\mu}+\frac{\theta^{\rm NC}}{3}V^{\mu}.
\end{equation}
The coefficient of $\kappa^{\rm NC}$ characterizes the component of the commutator parallel to the tangent $U^\mu$, whereas $\theta^{\rm NC}$ characterizes its component along the separation direction.
\\
\\
As highlighted in Section\ref{Section 2}, deviation vectors that differ only by a multiple of $U^\mu$ represent a displacement to the same nearby geodesic.  The physically interesting quantity is really the equivalence class of deviation vectors, where two deviation vectors are considered equivalent if their difference
is a multiple of $U^\mu$. We may therefore choose a representative in which the longitudinal contribution proportional to $U^\mu$ is removed. At the level of the above decomposition, this amounts to
\begin{equation}
\kappa_{\rm NC}U^\mu+\frac{\theta^{\rm NC}}{3}V^\mu
\longrightarrow
\kappa_{\rm NC}U^\mu+\frac{\theta^{\rm NC}}{3}V^\mu-\kappa_{\rm NC}U^\mu
=\frac{\theta^{\rm NC}}{3}V^\mu .
\end{equation}
Thus, adopting this representative for the separation vector, we may write
\begin{equation}
    [U,V]^{\mu}=\frac{\theta^{\rm NC}}{3}V^{\mu}.
\end{equation}
This form isolates the component of the non-commutativity that acts directly in the separation direction and is the quantity relevant for the transverse evolution of the congruence.
\\
\\
Such a Non-Commutativity has a particularly nice interpretation of preventing a Parallelogram formed from discrete $U$ and $V$ steps from closing. In particular, suppose $U$ and $V$ generate finite translations,
\begin{equation}
    T_U=e^{\epsilon U},\space T_V =e^{\epsilon V}.
\end{equation}
Now,
\begin{equation}
    T_U T_V=e^{\epsilon U}e^{\epsilon V}=e^{\epsilon V}e^{\epsilon U}e^{\epsilon[U,V]\dots}=T_V T_Ue^{\epsilon[U,V]\dots},
\end{equation}
so the failure of the elementary parallelogram to close is controlled at leading order by
\begin{equation}
    [U,V]=\frac{\theta^{\rm NC}}{3} V.
\end{equation}
This gives a nice geometric interpretation: discrete non-closure of the $U$- and $ V$-steps renders non-commutativity.
\\
\\
Let us revisit the justification for the congruence being hypersurface orthogonal. We now have
\begin{equation}
    U^\rho \nabla_\rho(U_\mu V^\mu)=U^\rho U_\mu \nabla_\rho V^\mu =U_\mu \mathcal{L}_U V^\mu+V^\rho U_\mu \nabla_\rho U^\mu=\frac{\theta^{\rm NC}}{3}U_{\mu}V^{\mu} .
\end{equation}
Hence, if $U_\mu V^\mu=0$ at any one point of the geodesic, uniqueness of the solution implies $U_\mu V^\mu=0$ everywhere. Thus, once the separation vector is chosen orthogonal to $U^\mu$ at one point, the modified commutator preserves this orthogonality along the congruence.
\\
\\
In terms of the projection tensor, the commutator is given by
\begin{equation}
    [U,V]^{\mu}=\frac{\theta^{\rm NC}}{3}V^{\mu}=\frac{\theta^{\rm NC}}{3}g^{\mu}_{\nu}V^{\nu}=\frac{1}{3}(\theta^{\rm NC}P^{\mu}_{\nu}V^{\nu}-\theta^{\rm NC}U^{\mu}U_{\nu}V^{\nu})=\frac{\theta^{\rm NC}}{3}P^{\mu}_{\nu}V^{\nu}.
\end{equation}
The evolution of the separation vector along the congruence is therefore modified according to
\begin{equation}
    \frac{DV^{\mu}}{d\tau}=U^{\nu}\nabla_{\nu}V^{\mu}=V^{\nu}\nabla_{\nu}U^{\nu}+[U,V]^{\mu}=V^{\nu}B^{\mu}_{\nu}+V^{\nu} \frac{\theta^{\rm NC}}{3}P^{\mu}_{\nu}
\end{equation}
Thus, the tensor $B^{\mu}_{\nu}=\nabla_{\nu}U^{\mu}$, which ordinarily governs the relative evolution of neighboring geodesics, is no longer sufficient to completely characterize the kinematics of the congruence. The additional $\theta^{\rm NC}$ dependence represents a non-commutative contribution to the evolution of the separation vector.
\\
\\
This naturally motivates the introduction of an effective deformation tensor,
\begin{equation}
\tilde{B}^{\mu}_{\nu}=B^{\mu}_{\nu}+ \frac{\theta^{\rm NC}}{3}P^{\mu}_{\nu}
\end{equation}
which governs the relative evolution of neighboring geodesics in the presence of non-commutativity.
\\
\\ 
The deformation tensor $B^{\mu}_{\nu}$ admits the standard irreducible decomposition into expansion, shear, and rotation,
\begin{equation}
    B_{\mu\nu}=\frac{1}{3}\theta P_{\mu\nu}+\sigma_{\mu\nu}+\omega_{\mu\nu},
\end{equation}
where $P^{\mu}_{\nu}=\delta^{\mu}_{\nu}+U^{\mu}U_{\nu}$ is the projection tensor onto the three-dimensional hypersurface orthogonal to the congruence,
$\theta$ is the expansion scalar $\sigma_{\mu\nu}$ is the symmetric traceless shear tensor, and $\omega_{\mu\nu}$ is the antisymmetric rotation tensor.
\\
\\
The effective deformation tensor therefore takes the form
\begin{equation}
    \tilde{B}_{\mu \nu}=\frac{1}{3}(\theta+\theta^{\rm{NC}})P_{\mu\nu}+\sigma_{\mu\nu}+\omega_{\mu\nu}.
\end{equation}
Consequently, the physical expansion, shear, and rotation of the congruence are given by
\begin{equation}
    \tilde{\theta}=\theta+\theta^{\rm NC}.
\end{equation}
The Raychaudhuri equation then reads
\begin{equation}
    \frac{d\tilde{\theta}}{d\tau}=-\frac{(\tilde{\theta}-\theta^{\rm{NC}})^2}{3}-R_{\mu\nu}U^{\mu}U^{\nu}-\sigma_{\mu\nu}\sigma^{\mu\nu}+\frac{d\theta^{\rm NC}}{d\tau}.
\end{equation}
This is mathematically the same as the standard Raychaudhuri equation except now, the expansion og the geodesic congruence is better interpreted by $\tilde \theta^{\rm NC}$ than $\theta$.
However, the equations for the shear $\sigma_{\mu\nu}$ and $\omega_{\mu\nu}$ retain the same explicit form and physical interpretation,
\begin{equation}
   U^\rho \nabla_\rho \sigma_{\mu\nu}=-\frac{2}{3}\theta \sigma_{\mu\nu}-\sigma_{\mu\rho}\sigma^{\rho}_{\nu}-\omega_{\mu\rho}\omega^{\rho}_{\nu}+\frac{1}{3}h_{\mu\nu}(\sigma_{\rho\sigma}\sigma^{\rho\sigma}-\omega_{\rho\sigma}\omega^{\rho\sigma})+C_{\rho\nu\mu\sigma}U^\rho U^\sigma +\frac{1}{2}\tilde R_{\mu\nu},
\end{equation}
\begin{equation}
    U^{\rho}\nabla_{\rho}\omega_{\mu\nu}=-\frac{2}{3}\theta\omega_{\mu\nu}-2\sigma^\rho_{[\nu}\omega_{\mu]\rho}.
\end{equation}
Now, suppose first that the strong energy condition holds,
\begin{equation}
    R_{ab}U^aU^b\geq 0,
\end{equation}
The modified Raychaudhuri equation then gives
\begin{equation}
    \frac{d\tilde{\theta}}{d\tau}+\frac{(\tilde{\theta}-\theta^{\rm{NC}})^2}{3}-\frac{d\theta^{\rm NC}}{d\tau}\leq 0,
\end{equation}
Expanding the square and, dividing by $\tilde\theta^2$, we obtain
\begin{equation}
    -\frac{d\tilde{\theta}^{-1}}{d\tau}+\frac{1}{3}-\frac{2\theta^{\rm NC}}{3\tilde \theta}-\frac{1}{\tilde{\theta}^2}\frac{d\theta^{\rm NC}}{d\tau}\leq 0.
\end{equation}
Hence,
\begin{equation}
    \tilde\theta^{-1}\geq \tilde\theta_0^{-1}+\frac{\tau}{3}-\int\left[\frac{2}{3}\frac{\theta^{\rm NC}}{\tilde \theta}+\frac{1}{\tilde{\theta}^2}\frac{d\theta^{\rm NC}}{d\tau}\right]d\tau.
\end{equation}
For a generic $\theta^{\rm NC}$, the integral on the right-hand side does not have a definite sign, and therefore the strong energy condition alone does not directly yield a focusing bound for $\tilde\theta\rightarrow -\infty$. 
\\
\\
However, it is important to distinguish this from the focusing of the standard expansion considered in the commutative case,
\begin{equation}
    \theta=\tilde\theta-\theta^{\rm NC}.
\end{equation}
Indeed, under the strong energy condition, the modified Raychaudhuri equation can still be written in its usual form, and we have
\begin{equation}
    \frac{d\theta}{d\tau}+\frac{\theta^2}{3}\leq 0,
\end{equation}
and therefore the standard focusing theorem applies. Thus, if $\theta(\tau_0)=\theta_0<0$ then $\theta\rightarrow -\infty $ within proper time$\tau\leq 3/|\theta_0|$. This also gives a sufficient condition for focusing of the modified expansion. If $\theta^{\rm NC}\in [-\infty,\infty)$ (i.e., it either takes negative values or finite positive values) along the geodesic up to the focusing point of $\theta$, then
\begin{equation}
    \tilde \theta=\theta+\theta^{\rm NC}\rightarrow -\infty
\end{equation}
within the same proper time. Thus, in this case, the usual strong energy
condition is sufficient to establish focusing of the modified expansion.
\\
\\
More generally, however, one may impose a modified energy condition, which we dub the timelike Non-Commutative Energy Condition (NCEC),
\begin{equation}\label{Non conummative energy condition}
    R_{\mu\nu}U^{\mu}U^{\nu}=8\pi\left[T_{\mu\nu}U^\mu U^\nu -\frac{1}{2}T\right]\geq \mathcal{E}^{\rm NC},
\end{equation}
where
\begin{equation}
    \mathcal{E}^{\rm NC}=\frac{2\theta^{\rm NC}\tilde\theta}{3}+\frac{d\theta^{\rm NC}}{d\tau}-\frac{(\theta^{\rm NC})^2}{3}.
\end{equation}
The significance of this condition depends on the sign of $\mathcal{E}_{\rm NC}$. If $\mathcal{E}_{\rm NC}>0$, it is stronger than the
usual strong energy condition, while if $\mathcal{E}_{\rm NC}=0$ it coincides
with it. Most importantly, if
\begin{equation}
    \mathcal{E}_{\rm NC}<0,
\end{equation}
then \eqref{Non conummative energy condition} is weaker than the usual strong
energy condition and may therefore be satisfied even when
\begin{equation}
    R_{\mu\nu}U^\mu U^\nu<0.
\end{equation}
Nevertheless, it is precisely sufficient to guarantee focusing of
$\tilde{\theta}$. In fact, substituting \eqref{Non conummative energy condition}
into the modified Raychaudhuri equation gives
\begin{equation}
    \frac{d\tilde{\theta}}{d\tau}
    +\frac{\tilde{\theta}^2}{3}\leq0.
\end{equation}
Thus, in the case $\mathcal{E}_{\rm NC}<0$, this condition provides a more
general focusing criterion than the argument based on the strong energy condition together with the finiteness of $\theta^{\rm NC}$. It establishes focusing of $\tilde{\theta}$ even in situations in which the usual strong energy condition is violated.
\\
\\
These observations may be summarized in the following lemma.
\begin{lemma}\label{Focusing within proper time with noncommutative correction timelke}
Let $U^\mu$ be the tangent field of a hypersurface-orthogonal timelike
geodesic congruence. Then $\tilde{\theta}$ goes to $-\infty$ along the geodesic within finite proper
time if either of the following conditions is satisfied:
\begin{itemize}
    \item $R_{\mu\nu}U^\mu U^\nu\geq0$, the corrected expansion $\tilde \theta_0$ is less that $\theta^{\rm NC}_0$, $\tilde \theta_0<\theta^{\rm NC}_0$ at some point on the geodesic, and
    $\theta^{\rm NC}$ remains finite along the geodesic up to the focusing
    point. In this case,
    \begin{equation}
        \tilde{\theta}\rightarrow-\infty
        \qquad\text{within}\qquad
        \tau\leq\frac{3}{|\tilde \theta_0-\theta^{\rm NC}_0|}.
    \end{equation}

    \item The timelike NCEC,
    \begin{equation}
        R_{\mu\nu}U^\mu U^\nu
        \geq
        \frac{2\theta^{\rm NC}\tilde{\theta}}{3}
        +\frac{d\theta^{\rm NC}}{d\tau}
        -\frac{(\theta^{\rm NC})^2}{3}
    \end{equation}
    holds, and the expansion $\tilde \theta^{\rm NC}$ takes a negative value
    $\tilde{\theta}_0<0$ at some point on the geodesic. In this case,
    \begin{equation}
        \tilde{\theta}\rightarrow-\infty
        \qquad\text{within}\qquad
        \tau\leq\frac{3}{|\tilde{\theta}_0|}.
    \end{equation}
\end{itemize}
\end{lemma}
Note that for a null geodesic congruence, the null-like NCEC instead has an $\mathcal{E}^{\rm NC}$ given by
\begin{equation}
    \mathcal{E}^{\rm NC}=\theta^{\rm NC}\tilde\theta+\frac{d\theta^{\rm NC}}{d\lambda}-\frac{(\theta^{\rm NC})^2}{2}.
\end{equation}
Repeating the same arguments as above, we have,
\begin{lemma}\label{Focusing within proper time with noncommutative correction null}
Let $k^\mu$ be the tangent field of a hypersurface-orthogonal null 
geodesic congruence.
Then $\tilde{\theta}$ goes to $-\infty$ along the geodesic within finite proper
time if either of the following conditions is satisfied:

\begin{itemize}
    \item $R_{\mu\nu}k^\mu k^\nu\geq0$, the expansion $\tilde \theta_0$ is less than the $\theta^{\rm NC}_0$, $\tilde\theta_0<\theta^{\rm NC}_0$ at some point on the geodesic, and
    $\theta^{\rm NC}$ remains finite along the geodesic up to the focusing
    point. In this case,
    \begin{equation}
        \tilde{\theta}\rightarrow-\infty
        \qquad\text{within}\qquad
        \lambda\leq\frac{2}{|\tilde\theta_0-\theta^{\rm NC}_0|}.
    \end{equation}

    \item The null-like NCEC,
    \begin{equation}
        R_{\mu\nu}k^\mu k^\nu
        \geq
        \theta^{\rm NC}\tilde{\theta}
        +\frac{d\theta^{\rm NC}}{d\lambda}
        -\frac{(\theta^{\rm NC})^2}{2}
    \end{equation}
    holds, and the modified expansion takes a negative value
    $\tilde{\theta}_0<0$ at some point on the geodesic. In this case,
    \begin{equation}
        \tilde{\theta}\rightarrow-\infty
        \qquad\text{within}\qquad
        \lambda\leq\frac{2}{|\tilde{\theta}_0|}.
    \end{equation}
\end{itemize}
\end{lemma}
As before, let us write the separation vector $V^\mu$ as
\begin{equation}
    V^\mu (\tau)=A^\mu_\nu (\tau)\frac{dV^\mu}{d\tau}(0).
\end{equation}
In matrix notation, we now obtain
\begin{equation}
    \frac{dA}{d\tau}=\tilde B A.
\end{equation}
Consequently, 
\begin{equation}
\tilde \theta= \text{Tr} \tilde B=\text{Tr}\left[\frac{dA}{d\tau} A^{-1}\right]=\frac{d}{d\tau}(\ln |\det A|).
\end{equation}
There, now, a necessary and sufficient condition for $q$ to be conjugate to $p$ is indeed that for the congruence of all timelike geodesics emanating from $p$, we have $\tilde \theta\rightarrow -\infty$ at $q$.
\begin{prop}
    Let ($\mathcal{M}$,$g_{\mu\nu}$) be a spacetime satisfying $R_{\mu\nu}U^\mu U^\nu\geq0$ for all timelike $U^\mu$. Let $\gamma$ be a timelike geodesic and let $p\in \gamma$. Suppose the convergence of the congruence of timelike geodesics emanating into the future from $p$ is less than $\theta^{\rm NC}_0$, $\tilde \theta_0<\theta^{\rm NC}_0$ at $r\in\gamma$, and $\theta^{\rm NC}$ remains finite along the geodesic up to the focusing point. Then within proper time $\tau\leq 3/|\tilde\theta_0-\theta^{\rm NC}_0|$ from $r$ along $\gamma$ there exists a point $q$ conjugate to $p$, assuming, of course, that $\gamma$ extends that far.
\end{prop}
\begin{prop}
    Let ($\mathcal{M}$,$g_{\mu\nu}$) be a spacetime satisfying $R_{\mu\nu}U^\mu U^\nu\geq \mathcal{E}^{\rm NC}$ for all timelike $U^\mu$. Let $\gamma$ be a timelike geodesic and let $p\in \gamma$. Suppose the convergence of the congruence of timelike geodesics emanating into the future from $p$ attains the value $\tilde \theta_0<0$ at $r\in\gamma$. Then within proper time $\tau\leq 3/|\tilde\theta_0-\theta^{\rm NC}_0|$ from $r$ along $\gamma$ there exists a point $q$ conjugate to $p$, assuming, of course, that $\gamma$ extends that far.
\end{prop}
Note that in the noncommutative case, the geodesic deviation equation is given by
\begin{equation}
 U^\nu \nabla_\nu (U^\lambda \nabla_\lambda V^\mu)=-R^\mu_{\nu \lambda \rho}V^\mu U^\lambda U^\rho+\frac{1}{3}U^\nu \nabla_\nu (\theta^{\rm NC}V^\mu).
\end{equation}
Let us now turn our attention to the extremal length properties of timelike geodesics. We have
\begin{equation}
\begin{split}
    \frac{d\tau}{d\alpha}=\int_a^b \frac{\partial f}{\partial \alpha} dt\\
    =\int_a^b X^\mu\nabla_\mu (-T^\nu T_\nu)^{1/2}dt\\
    =-\int_a^b \frac{1}{f}T_\nu X^\mu \nabla_\mu T^\nu dt \\
    = -\int_a^b \frac{1}{f}T_\nu T^\mu \nabla_\mu X^\nu dt+\frac{1}{3}\int_a^b\frac{1}{f}T_\nu X^\nu \theta^{\rm NC}\\=-\int_a^b T^\mu\nabla_\mu\left[\frac{1}{f}T_\nu X^\nu\right]dt +\int_a^b X^\nu T^\mu\nabla_\mu (T_\nu /f) dt+0\\=\int_a^b X^\nu T^\mu \nabla_\mu(T_\nu /f) dt,
\end{split}
\end{equation}
The second variation of arc length is
\begin{equation}
    \frac{d\tau^2}{d\alpha^2}=\int_a^b X^\rho \nabla_{\rho} [X^\nu T^\mu\nabla_\mu(T_\nu/f)]dt.
\end{equation}
Evaluating this expression at $\alpha=0$, assuming that $\lambda_0$ is a geodesic, we find
\begin{equation}\label{proper time second deivative 2}
    \begin{split}
        \frac{d^2\tau}{d\alpha^2}\bigg|_{\alpha=0}=\int_a^b X^\nu (X^\rho \nabla_\rho T^\mu)\nabla_\mu (T_\nu/f) dt +\int_a^b X^\nu T^\mu X^\rho \nabla_\rho \nabla_\mu (T_\nu /f) dt \\
        =\int_a^b X^\nu (T^\rho \nabla_\rho X^\mu)\nabla_\mu (T_\nu /f) dt +\int_a^b X^\nu T^\mu X^\rho \nabla_\mu \nabla_\rho (T_\nu/f) dt+\int_a^b \frac{X^\nu T^\mu X^\rho R_{\rho \mu \nu}^\sigma T_\sigma}{f}dt \\
        =\int_a^b X^\nu T^\rho\nabla_\rho [X^\mu\nabla_\mu (T_\nu/f)] dt+\int_a^b \frac{X^\nu T^\mu X^\rho R_{\rho\mu\nu}^{\sigma} T_{\sigma}}{f}f dt.
    \end{split}
\end{equation}
The term in the square brackets can be reexpressed as
\begin{equation}
\begin{split}
    X^\mu\nabla_\mu (T_\nu/f)=\frac{1}{f}X^\mu\nabla_\mu T_\nu -\frac{1}{f^2}T_\nu X^\mu\nabla_\mu f\\
    =\frac{1}{f}T^\mu\nabla_\mu X_\nu-\frac{1}{f}\theta^{\rm NC}X_\nu +\frac{1}{f^2}T_\nu T^\mu \nabla_\mu\left[\frac{1}{f}T_\rho X^\rho\right].
\end{split}
\end{equation}
Eqn.\eqref{proper time second deivative 2} can be put in a more easily recognizable form by choosing the curve parametrization so that $f=1$ along the geodesic $\lambda_0$ and also choosing the deviation vector $X^\mu$ to be orthogonal to $T^\mu$ along $\lambda_0$. With these choices, we find
\begin{equation}
    \begin{split}
         \frac{d^2\tau}{d\alpha^2}\bigg|_{\alpha=0}=\int_a^b X^\nu\{T^\rho\nabla_\rho(T^\mu\nabla_\mu X_\nu)-\frac{1}{3}T^\rho \nabla_\rho (\theta^{\rm NC} X_\nu)+R_{\rho\mu\nu}^{\sigma}T^\mu T_{\sigma}X^\rho \} dt\\
         =\int_a^b X^\nu (\mathcal{O}X)_\nu dt,
    \end{split}
\end{equation}
and we have, as before,
\begin{thm}
    Let $\gamma$ be a smooth timelike curve connecting two points $p,q\in \mathcal{M}$. Then the necessary and sufficient condition that $\gamma$ locally maximizes the proper time between $p$ and $q$ over smooth one-parameter variations is that $\gamma$ be a geodesic with no point conjugate to $p$ between $p$ and $q$.
\end{thm}
Now, the extrinsic curvature $K$ of a spatial hypersurface $\Sigma$ is given by
\begin{equation}
    K=\theta=\tilde\theta-\theta^{\rm NC}.
\end{equation}
Analogous to Proposition \ref{focsuing extrinsic curvature}, we now have
\begin{lemma}
Let $U^\mu$ be the tangent field of a hypersurface-orthogonal timelike
geodesic congruence. Let $\Sigma$ be a spacelike hypersurface with extrinsic curvature $K$. Then there exists a point $p$ conjugate to $\Sigma$ along the geodesic $\gamma$ orthogonal to $\Sigma$ and passing through $q$, if either of the following conditions is satisfied:
\begin{itemize}
    \item $R_{\mu\nu}U^\mu U^\nu\geq0$, $K<0$   and
    $\theta^{\rm NC}$ remains finite along the geodesic up to the focusing
    point. In this case, the point $p$ conjugate to $\Sigma$ exists within $\tau\leq 3/|K|$
    \item The timelike NCEC,
    \begin{equation}
        R_{\mu\nu}U^\mu U^\nu
        \geq
        \frac{2\theta^{\rm NC}\tilde{\theta}}{3}
        +\frac{d\theta^{\rm NC}}{d\tau}
        -\frac{(\theta^{\rm NC})^2}{3},
    \end{equation}
    holds, $K<-\theta^{\rm NC}$. In this case, the point $p$ conjugate to $\Sigma$ exists within $\tau\leq 3/|K+\theta^{\rm NC}|$
\end{itemize}
\end{lemma}
Similarly, for null geodesics, we have
\begin{lemma}
Let $k^\mu$ be the tangent field of a hypersurface-orthogonal null
geodesic congruence. Let $S$ be a smooth two-dimensional spacelike submanifold. Then there exists a point $p$ conjugate to $S$ along the outgoing null geodesic $\mu$ passing through $p$, if either of the following conditions is satisfied:
\begin{itemize}
    \item $R_{\mu\nu}k^\mu k^\nu\geq0$, the corrected expansion $\tilde \theta_0$ is less that $\theta^{\rm NC}_0$, $\tilde \theta_0<\theta^{\rm NC}_0$ at some point on the geodesic, and
    $\theta^{\rm NC}$ remains finite along the geodesic up to the focusing
    point. In this case, the point $p$ conjugate to $S$ exists within $\tau\leq 2/|\tilde \theta-\theta^{\rm NC}|$
    \item The null-like NCEC,
    \begin{equation}
        R_{\mu\nu}k^\mu k^\nu
        \geq
        \theta^{\rm NC}\tilde{\theta}
        +\frac{d\theta^{\rm NC}}{d\lambda}
        -\frac{(\theta^{\rm NC})^2}{2}
    \end{equation}
    holds, and the expansion $\tilde \theta^{\rm NC}$ takes a negative value
    $\tilde{\theta}_0<0$ at some point on the geodesic. In this case, the point $p$ conjugate to $S$ exists within $\lambda\leq 2/|\tilde \theta|$
\end{itemize}
\end{lemma}
These local focusing results can then be incorporated into the standard global arguments underlying the singularity theorems. Once every past-directed timelike geodesic orthogonal to the Cauchy surface develops a conjugate point within a uniformly bounded proper time, the usual maximality argument implies that no such geodesic can be extended indefinitely into the past. The resulting singularity theorems therefore differ from their classical counterparts only through the modified focusing condition and the corresponding bound on the proper time to conjugacy.
\begin{thm}\label{corrected singularity theorem 1}
Let $(\mathcal{M},g_{\mu\nu})$ be a globally hyperbolic spacetime. Suppose there exists a smooth (or, at least $C^2$) spacelike Cauchy surface $\Sigma$. Then, all past-directed timelike geodesics are incomplete if
\begin{itemize}
    \item The Strong Energy Condition, $R_{\mu\nu}U^\mu U^\nu\geq0$ holds and the extrinsic curvature of the Cauchy surface $\Sigma$ satisfies $K\leq C<0$ everywhere, and
    $\theta^{\rm NC}$ remains finite along the geodesic up to the focusing
    point. In this case, no past timelike directed curve from $\Sigma$ can have length greater than $3/|C|$
    \item The timelike NCEC,
    \begin{equation}
        R_{\mu\nu}U^\mu U^\nu
        \geq
        \frac{2\theta^{\rm NC}\tilde{\theta}}{3}
        +\frac{d\theta^{\rm NC}}{d\tau}
        -\frac{(\theta^{\rm NC})^2}{3}
    \end{equation}
    holds, and the extrinsic curvature of the Cauchy surface $\Sigma$ satisfies $K\leq C<-\theta^{\rm NC}$ everywhere. In this case, no past timelike directed curve from $\Sigma$ can have length greater than $3/|C+\theta^{\rm NC}|$.
\end{itemize}
\end{thm}
\begin{thm}\label{corrected singularity theorem 2}
Let $(\mathcal{M},g_{\mu\nu})$ be a connected, globally hyperbolic spacetime with a noncompact Cauchy surface $\Sigma$. Suppose further that $\mathcal{M}$ contains a trapped surface $T$. Then, at least one inextensible future-directed orthogonal null geodesic from $T$ is incomplete if,
\begin{itemize}
    \item The Strong Energy Condition, $R_{\mu\nu}k^\mu k^\nu\geq0$ holds, the maximum value of $\tilde \theta-\theta^{\rm NC}$ satisfies $\tilde \theta_0-\theta^{\rm NC}<0$, and
    $\theta^{\rm NC}$ remains finite along the geodesic up to the focusing
    point. In this case, at least one inextensible future-directed orthogonal null geodesic from $T$ has affine length no greater than $2/|\tilde \theta_0-\theta_0^{\rm NC}|$
    \item The null-like NCEC,
    \begin{equation}
        R_{\mu\nu}k^\mu k^\nu
        \geq
        \theta^{\rm NC}\tilde{\theta}
        +\frac{d\theta^{\rm NC}}{d\lambda}
        -\frac{(\theta^{\rm NC})^2}{2}
    \end{equation}
    holds, the maximum value of $\tilde \theta$ satisfies $\tilde \theta_0<0$. In this case, at least one inextensible future-directed orthogonal null geodesic from $T$ has affine length no greater than $2/|\tilde \theta_0|$.
\end{itemize}
\end{thm}
The results of this section show that the relaxation of the commutativity condition between the tangent and separation vectors leads to a consistent deformation of the kinematics of geodesic congruences while preserving the underlying smooth spacetime geometry. The relevant expansion is shifted from $\theta$ to $\tilde \theta=\theta+\theta^{\rm NC}$ and the corresponding focusing and conjugacy properties are governed by this modified expansion. In particular, focusing can still follow from the Strong Energy Condition together with suitable finiteness of $\theta^{\rm NC}$ while the timelike and null-like NCEC introduced above provides a more general criterion that can remain sufficient even when the standard strong energy condition is violated. The resulting timelike and null focusing statements consequently lead to corresponding modifications of the usual geodesic incompleteness theorems. Having established these geometric consequences of non-commutativity, we now turn to their implications for homogeneous and isotropic cosmology, where the modified expansion is directly related to the Hubble parameter.
\section{Cosmology and Non-Commutativity} \label{Section 4}
To explore the cosmological implications of the modified congruence, let us now specialize to a homogeneous and isotropic FLRW spacetime. In this case, the expansion scalar is directly related to the Hubble
parameter through
\begin{equation}
    \tilde{\theta}=3H.
\end{equation}
Substituting this relation into the modified Raychaudhuri equation yields the corresponding evolution equation for the Hubble parameter,
\begin{equation}
    3\dot{H}=-3H^2+2H\theta^{\rm{NC}}-\frac{(\theta^{\rm NC})^2}{3}-R_{\mu\nu}U^{\mu}U^{\nu}+\dot{\theta}^{\rm{NC}}
\end{equation}
This equation shows that the non-commutative contribution modifies the evolution of the Hubble parameter through three distinct terms: one proportional to $H\theta_{\rm{NC}}$, a quadratic correction in $\theta_{\rm{NC}}$, and a dynamical contribution arising from $\dot{\theta}_{\rm{NC}}$. In this picture, the expansion of the Universe is no longer governed solely by the matter content and spacetime curvature, but also by the non-commutative deformation of the geodesic congruence.
\\
\\
The modified Hubble evolution equation may be rewritten in the more familiar form,
\begin{equation}
    \frac{\ddot{a}}{a}=-\frac{4\pi G}{3}(\rho+3P)+\frac{2}{3}H\theta^{\rm{NC}}-\frac{(\theta^{\rm{NC}})^2}{9}+\frac{\dot{\theta}^{\rm{NC}}}{3}.
\end{equation}
The additional terms may be interpreted as an effective geometric source of acceleration arising purely from the non-commutative structure of the congruence. Unlike conventional cosmological models, these corrections do not originate from introducing an additional exotic field or modifying Einstein's field equations. Instead, they emerge from relaxing the usual assumption that the tangent vector of a geodesic commutes with the separation vector connecting it to a neighboring geodesic.
\\
\\
The role of non-commutativity can be thought of effectively modifying the pressure as 
\begin{equation}\label{effective presuure}
    P_{\rm eff}=P-\frac{H\theta^{\rm NC}}{6\pi G}+\frac{(\theta^{\rm NC})^2}{36\pi G}-\frac{\dot{\theta}^{\rm NC}}{12\pi G},
\end{equation}
while the energy density (the $00$ component of the stress energy tensor) remains the same if the Friedmann equation is to remain the same.
\\
\\
The modified acceleration equation therefore admits a simple criterion for accelerated expansion:
\begin{equation}
    +\frac{2}{3}H\theta^{\rm{NC}}- \frac{(\theta^{\rm{NC}})^2}{9}+\frac{\dot{\theta}^{\rm{NC}}}{3}>\frac{4\pi G}{3}(\rho+3P)
\end{equation}
Whenever this condition is satisfied, the repulsive contribution arising from the non-commutative deformation exceeds the attractive gravitational effect of ordinary matter and pressure, leading to
\begin{equation}
    \ddot{a}>0
\end{equation}
This raises the intriguing possibility that non-commutativity in the kinematics of geodesic congruences may provide a geometric mechanism capable of driving the observed accelerated expansion of the Universe. In this interpretation, the accelerated expansion is not attributed to an exotic dark-energy component, but rather to the modified evolution of spacetime congruences induced by non-commutativity.
\\
\\
Combining effective pressure in eqn.\eqref{effective presuure}, with the conservation of the stress energy tensor,
\begin{equation}
    \nabla_{\mu}T^{\mu\nu}=0,
\end{equation}
one finds that the modified continuity equation is given by,v
\begin{equation}\label{modified coninutiy}
\dot{\rho}=-3H(\rho+P)+\frac{4\rho}{3}\theta^{\rm{NC}}-\frac{H(\theta^{\rm{NC}})^2}{12\pi G}+\frac{H\dot{\theta}^{\rm{NC}}}{4\pi G}.
\end{equation}
Eqn.\eqref{modified coninutiy} demonstrates that the standard continuity equation is fundamentally modified by the non-commutative structure of spacetime. In conventional FLRW cosmology, the conservation of the cosmic energy-momentum tensor follows directly from the Bianchi identities and governs the dilution of the cosmological fluids solely through the Hubble expansion. In the present framework, however, the relaxation of the commutativity condition between the tangent and separation vectors introduces additional geometric contributions to the conservation law, which can be thought of as arising from the effective pressure defined in eqn.\ref{effective presuure}. Consequently, the evolution of the cosmic energy density is determined not only by the expansion of the Universe but also by the non-commutative expansion scalar $\theta_{\rm{NC}}$.
\\
\\
A notable feature of eqn.\eqref{modified coninutiy} is that the non-commutative corrections naturally separate into two distinct classes. The first is proportional to the energy density itself through the term $4\rho\theta^{\rm NC}/3$, while the remaining contributions is non-linearly proportional to the total energy density (since $H\propto \sqrt{\rho}$) and therefore cannot be decomposed into separate terms proportional to the individual energy density components. Since the former scales with the corresponding fluid density, it is natural to interpret it as a universal geometric modification to the dilution of every cosmological component. The terms that aren't linearly proportional to the density are however fundamentally different and cannot be attributed to any conventional cosmological fluid.
\\
\\
Motivated by this distinction, we assume that the universal correction proportional to $4\rho \theta^{\rm NC}/3$ modifies the evolution of all cosmological fluids, whereas the terms that are non-linear in energy density are absorbed into an effective dark-energy sector. This prescription preserves the standard physical interpretation of matter and radiation while identifying the genuinely new contribution generated by non-commutativity with an effective geometric dark-energy component.
\\
\\
The continuity equation for individual cosmological fluid therefore becomes
\begin{equation}\label{individual continuity}
\dot{\rho}_i+3H(1+\omega_i)\rho_i=\frac{4\theta^{\rm NC}\rho_i}{3},
\end{equation}
where $i=r,m,k$ (radiation, matter and curvature). The solution of eqn.\eqref{individual continuity} is given by 
\begin{equation}\label{individual evolution}
\rho_i(a)=\rho_{i,0}a^{-3(1+\omega_i)}\text{exp}\left[\int_{t_0}^t\frac{4}{3}\theta^{\rm NC}(t')dt'\right]
\end{equation}
Eqn.\eqref{individual evolution} demonstrates that the standard cosmological scaling laws are modified by the non-commutative geometry. In particular, 
\begin{equation}
    \rho_{\rm r}=\rho_{\rm r,0}a^{-4}\text{exp}\left[\int \frac{4}{3}\theta^{\rm NC}dt\right], \quad \rho_{\rm m}=\rho_{\rm m,0}a^{-3}\text{exp}\left[\int \frac{4}{3}\theta^{\rm NC}dt'\right],
\end{equation}
while the curvature contribution evolves according to
\begin{equation}\label{modified curvature evolution}
    \rho_{\rm k}=\rho_{\rm k,0}a^{-2}\text{exp}\left[\int\frac{4}{3}\theta^{\rm NC}dt'\right].
\end{equation}
It is convenint to define 
\begin{equation}
    F=\text{exp}\left(-\int\frac{4}{3}\theta^{\rm NC}dt'\right).
\end{equation}
Now, during radiation domination,
\begin{equation}
    H^2\simeq \frac{8\pi G}{3}\rho_{\rm r},
\end{equation}
so that greater values of $\theta^{\rm NC}$ at early times enhance the radiation energy density, hereby increasing the early-time Hubble parameter. Since the CMB tightly constrains the angular acoustic scale \cite{2020},
\begin{equation}
      \theta_s=\frac{r_s(a_*)}{\chi(a_*)}=\frac{\int_{0}^{a_*}\frac{c_s}{a^2H(a)} da}{\int_{a_*}^{a_0} \frac{da}{a^2H(a)}},
\end{equation}
an enhanced early expansion rate decreases the sound horizon $r_s$. To preserve the observed value of $\theta_s$, this reduction must be accompanied by a corresponding decrease in the comoving distance to last scattering $\chi_s(a_*)$ which generally requires a larger inferred value of the present-day Hubble constant $H_0$. Consequently, values of $F<1$ before recombination provide a mechanism capable of shifting the CMB-inferred value of $H_0$ toward that measured from late-time observations, thereby improving the Hubble tension.
\\
\\
On the other hand, the dark-energy sector, exhibits a rich dynamics. Assuming $\rho=-P$ (although, $\rho\not= P_{\rm eff}$), its continuity equation is given by
\begin{equation}\label{dark energy evolution}
    \dot{\rho}_{\rm DE}-\frac{4\rho_{\rm DE}\theta^{\rm NC}}{3}+\frac{H(\theta^{\rm NC})^2}{12\pi G}-\frac{H\dot{\theta}^{\rm NC}}{4\pi G}=0.
\end{equation}
which differs fundamentally from the evolution equations for matter, radiation and curvature. Whereas these components receive only the universal multiplicative correction proportional to $\theta^{\rm NC}$, the effective dark-energy sector contains additional geometric terms generated directly by the non-commutative modification of the continuity equation. Consequently, the dark-energy evolution is not merely a rescaled version of the standard conservation law but instead acquires genuinely new dynamics determined by the underlying non-commutative geometry.
\\
\\
Eqn.\eqref{dark energy evolution} is a first-order linear differential equation and therefore admits the formal solution
\begin{equation}\label{de}
\begin{split}
    \rho_{\rm DE}=\text{exp}\left(\int_{t_0}^t\frac{4\theta^{\rm NC}}{3}dt'\right)\left[\rho_{\rm DE}(t_0)-\frac{1}{12\pi G}\int_{t_0}^t\text{exp}\left(-\int_{t_0}^{t}\frac{4\theta^{\rm NC}}{3}dt'\right)((\theta^{\rm NC})^2-3\dot{\theta}^{\rm NC})Hdt'\right]\\
    =\frac{1}{F}\left[\rho_{\rm DE}(t_0)+\frac{3}{16\pi G}\int_{t_0}^t H\left(\frac{3\dot{F}^2}{4F}-\ddot{F}\right)dt'\right]
\end{split}
\end{equation}
Unlike the matter, radiation and curvature sectors, whose evolution is determined entirely by the exponential correction arising from the modified dilution law, the dark-energy density depends explicitly upon both the non-commutative expansion scalar and its time derivative through the integral term. Consequently, the effective dark-energy sector possesses a considerably richer phenomenology than the remaining cosmological components. The first exponential factor represents the universal modification of cosmological dilution produced by non-commutativity, whereas the integral term continuously sources or depletes the dark-energy density according to the evolution of the non-commutative geometry itself.
\\
\\
Comparing with $\rho_{\rm DE}=\rho(t_0)_{\rm DE}\exp\left({\int-3(1+\omega_{\rm DE})Hdt}\right)$, one finds that  
\begin{equation}\label{eos}
    \omega_{\rm DE}=-1+\frac{\dot{F}}{3HF}+\frac{4\ddot{F}-\frac{3\dot{F}^2}{F}}{64\pi \rho_{\rm DE}(t_0)+12\int_{t_0}^tH\left(\frac{3\dot{F}^2}{4F}-\ddot{F}\right)dt'}.
\end{equation}
This result demonstrates that dynamical dark energy \cite{Ratra:1987rm,Wetterich_1988,Caldwell_2002,COPELAND_2006} emerges naturally as a purely geometric consequence of the non-commutative structure of spacetime. Unlike conventional implementations of dynamical dark energy, in which the evolution of the equation of state is introduced through some additional field or just some phenomenological parametrization, here the time dependence of $\omega_{\rm DE}$ arises directly from the non-commutative deformation of geodesic congruences. This prediction is also in line with recent observational indications from DESI and complementary cosmological probes, which have found a mild preference for an evolving dark energy equation of state over a cosmological constant \cite{Adame_2025,Gialamas_2025,Giar_2025}. These observations therefore provide an interesting phenomenological context for the present geometric mechanism, in which dynamical dark energy arises completely through the underlying non-commutative nature of spacetime.
\\
\\
In this model, the final fate of the cosmos is entirely dependent on the parametrization of $\theta^{\rm NC}$.
\\
\\
In the epoch where the effective dark-energy density dominates the cosmic energy budget, the Friedmann equation reduces approximately to
\begin{equation}\label{firedmann de}
    H^2\simeq \frac{8\pi G}{3}\rho_{\rm DE}.
\end{equation}
Substituting eqn.\eqref{firedmann de} and solving for the scale factor, one obtains,
\begin{equation}\label{scale factor}
a=a_0\text{exp}\left(H_0\int_{t_0}^t\frac{dt'}{\sqrt{F}}\left[H_0+\frac{1}{4}\int_{t_0}^{t'}\frac{1}{\sqrt{F}}\left(\frac{3\dot{F}^2}{F}-\ddot{F}\right)dt''\right]\right)
\end{equation}
Eqn\eqref{scale factor} together with eqn.\eqref{eos} illustrates that the late-time expansion history is governed entirely by the functional form of the non-commutative expansion scalar. Different parameterizations of $F$ (and hence, $\theta_{\rm NC}$) therefore generate qualitatively different effective dark-energy sectors without requiring additional exotic fields or phenomenological equations of state. However, eqn.\eqref{eos} is manifestly nontrivial, depending on the non-commutative deformation through both its first and second time derivatives, as well as a quadratic dependence on the first derivative. Consequently, the conditions under which the effective dark-energy sector is canonical, phantom-like, or cosmological-constant-like are themselves nontrivial. 
\\
\\
Note during DE domination, that the condition $\omega=-1$ does not require $F(t)=1$, but rather any solution to the following non linear differential equation
\begin{equation}
4\ddot{F}-3\frac{\dot{F}^2}{F}+8H_0\dot{F}=0.
\end{equation}
A corresponding class of solutions is
\begin{equation}
    F(t)=\left[1+\frac{\dot{F}_0}{8H_0} \left(1-e^{-2H(t-t_0)}\right)\right]^4,
\end{equation}
any $F(t)$ belonging to this class yields a dark energy sector that mimicks a cosmological constant evolution despite $F(t)\not= 1$ in general.
\\
\\
More generally, the dark-energy sector exhibits phantom behavior, $\omega<-1$, when 
\begin{equation}
    4F\ddot{F}-3\dot{F}^2+8HF\dot{F}<0,
\end{equation}
whereas it exhibits canonical (non-phantom) behavior, $\omega>-1$ when
\begin{equation}
    4F\ddot{F}-3\dot{F}^2+8HF\dot{F}>0.
\end{equation}
Consequently, a phantom crossing occurs when $  4F\ddot{F}-3\dot{F}^2+8HF\dot{F}=0$.
\\
\\
These results demonstrate that the non-commutative deformation can lead to a rich variety of effective dark-energy behaviors.
\\
\\
Overall, non-commutativity plays a nontrivial role in the resulting cosmological dynamics and in the effective dark-energy sector, with implications for the possible late-time fate of the Universe. Depending on the functional form of the non-commutative deformation and the solutions of non-trivial differential equations such as \eqref{scale factor},\eqref{de}, and the resulting equation of state parameter can exhibit canonical, phantom-like, or cosmological-constant behavior.
\\
\\
Furthermore, one may show that other forms of non-commutativity, that generically modify the geodesic equations, such as the framework explored in \cite{Matt_Gregory_2026}, do not yield a dark energy sector with such rich modifications.  Let us briefly review their approach and analyze the resulting dark energy sector. They considered a non-commutative generalization of the Klein-Gordon equation given by
\begin{equation}
    \left(\hat\square+\frac{m^2c^2}{\hbar^2}\right)\Phi=\frac{1}{\sqrt{g}}\star [\partial_\nu [(\sqrt{g}g^{\mu\nu})\star \partial_\mu \Phi]]+\frac{m^2c^2}{\hbar^2}\Phi,
\end{equation}
where the $\star$-product between any two functions $f$ and $g$ is definied as 
\begin{equation}
  f\star g=fg+\frac{i}{2}\Theta^{\mu\nu}\partial_\mu f \partial_\nu g +\mathcal{O}(\Theta^2). 
\end{equation}
with$\Theta^{\mu\nu}$ being a  constant asymmetric tensor.
\\
\\
Such a non-commutativity modifies the energy-momentum dispersion relation to
\begin{equation}
    g^{\mu\nu}p_\mu p_\nu+M^2(x)=0,
\end{equation}
where the effective position-dependent mass function, $M(x)$, is given by
\begin{equation}
    M^2(x)=m^2c^2 \frac{\sqrt{g}_*}{\sqrt{g}}.
\end{equation}
This position-dependent mass function leads to the modified geodesic equation
\begin{equation}
\ddot{x}^{\mu}+\Gamma^\mu_{\alpha\beta}\dot{x}^\alpha\dot{x}^\beta=-\frac{1}{2}g^{\mu\nu}\frac{\partial M^2}{\partial x^\nu}.
\end{equation}
Expanding the effective mass squared gives
\begin{equation}
    M^2=m^2c^2\left(1+\frac{1}{2}\Theta^{\mu_1\nu_1}\Theta^{\mu_2\nu_2}(\partial_{\mu_1}\partial_{\mu_2}\sqrt{g})\left(\partial_{\nu_1}\partial_{\nu_2}\frac{1}{\sqrt{g}}\right)+\mathcal{O}(\Theta^4)\right).
\end{equation}
so that the geodesic equation becomes
\begin{equation}
    \ddot{x}^\mu+\Gamma^\mu_{\alpha\beta}\dot{x}^\alpha \dot{x}^\beta=-\frac{m^2c^2}{4}g^{\mu\nu}\partial_{\nu}\left[\Theta^{\mu_1\nu_1}\Theta^{\nu_1\nu_2}(\partial_{\mu_1}\partial_{\mu_2}\sqrt{g})\left(\partial_{\nu_1}\partial_{\nu_2}\frac{1}{\sqrt{g}}\right)\right]+\mathcal{O}(\Theta^4).
\end{equation}
In a homogeneous and isotropic FLRW background, the key non-commutative contraction simplifies to
\begin{equation}
    \Theta^{\mu_1\nu_1}\Theta^{\nu_1\nu_2}(\partial_{\mu_1}\partial_{\mu_2}\sqrt{g})\left(\partial_{\nu_1}\partial_{\nu_2}\frac{1}{\sqrt{g}}\right)= \Theta^{00}\Theta^{00}\frac{d^2a^3}{dt^2}\frac{d^2}{dt^2}\left(\frac{1}{a^3}\right).
\end{equation}
However, since $\Theta^{\mu\nu}$ is a constant asymmetric tensor, its temporal component identically vanishes, $\Theta^{00}=0$. Consequently, this form of non-commutativity induces no modifications to the geodesic equations, the Raychaudhuri equation, the continuity equations, or the dark energy equation of state parameter, leaving the entire dark energy sector completely untouched, at least up to order $\mathcal{O}(\Theta^2)$.
\\
\\
This crucial distinction is precisely what makes our kinematic approach to non-commutativity so compelling. By directly generalizing the commutation properties of separation and tangent vectors in geodesic congruences rather than modifying coordinate algebras or particle masses, our framework generates non-trivial geometric terms in the Raychaudhuri and continuity equations. This directly drives the dynamics of dark energy and produces a rich, evolving dark energy sector without requiring exotic matter fields or ad hoc modifications to Einstein's equations.
\section{Conclusions}\label{Section 5}
In this work we have proposed a novel geometric realization of spacetime non-commutativity that preserves the smooth differentiable manifold of General Relativity. Rather than introducing non-commutativity through the spacetime coordinates or replacing the underlying manifold by a non-commutative algebra, we considered the possibility that non-commutativity is encoded in the kinematics of neighboring geodesics. The construction is formulated within the standard differential-geometric framework of geodesic congruences and their associated kinematical quantities, while relaxing the usual commutativity assumption for the tangent and separation vectors \cite{Wald1984,oneill1983semiriemannian,beem2017global}.Motivated by the Frobenius theorem \cite{frobenius1877ueber,warner1983foundations}, the conventional assumption that the tangent and separation vectors commute was relaxed, leading naturally to a generalized deformation tensor and a modified Raychaudhuri equation \cite{1955PhRv...98.1123R,Raychaudhuri1957SingularSI}.
\\
\\
A central consequence of this deformed kinematics, is its effect on the classical singularity theorems \cite{PhysRevLett.14.57,PhysRevLett.17.444,Hawking:1970zqf,Hawking:1967ju,HawkingEllis1973,penrose1972techniques,beem2017global,Witten_2020,Senovilla_2022,Senovilla1998}. We found that the physically relevant expansion is shifted from $\theta$ to $\tilde \theta=\theta+\theta^{\rm NC}$, and that the usual strong and weak energy conditions, while still sufficient for focusing when $\theta^{\rm NC}$ remains finite along the geodesic, are no longer necessary. We introduced the timelike and null-like NCEC as a generalized focusing criterion and argued that it can hold even when the strong energy condition is violated. This yields modified versions of the Raychaudhuri-based focusing lemmas (Lemmas \ref{Focusing within proper time with noncommutative correction timelke} and \ref{Focusing within proper time with noncommutative correction null}), the conjugate-point propositions, and ultimately non-commutative generalizations of the two principal singularity theorems: incompleteness of past-directed timelike geodesics from a Cauchy surface of sufficiently negative extrinsic curvature (Theorem \ref{corrected singularity theorem 1}), and incompleteness of null geodesics orthogonal to a trapped surface (Theorem \ref{corrected singularity theorem 2}). In both cases the structure of the classical theorem is preserved, but the bound on the proper time or affine length to conjugacy, and the energy condition required to guarantee it, are modified by $\theta^{\rm NC}$ and its time derivative. 
\\
\\
Cosmologically, the modified Raychaudhuri equation induces a generalized continuity equation whose non-commutative corrections separate cleanly into two classes: a term proportional to the fluid density itself, which acts as a universal correction to the dilution of matter, radiation, and curvature; another term that is non-linearly proportional to the energy density, which we absorb into the dark-energy sector. Dynamical dark energy \cite{Ratra:1987rm,Wetterich_1988,Caldwell_2002,COPELAND_2006}, whose observational relevance is suggested by DESI \cite{Adame_2025,Gialamas_2025,Giar_2025}, therefore emerges directly from the non-commutative structure of the congruence, without an additional scalar field or modification to Einstein's equations. The entire late-time expansion history, including the equation of state $\omega_{\rm DE}$ is determined once a functional form for $\theta^{\rm NC}$ is specified, so different parameterizations generate qualitatively distinct dark-energy phenomenologies from a single geometric input. We also noted two observational anchor points: a positive $\theta^{\rm nc}$ at early times suppresses the sound horizon at recombination, offering a geometric route to reconciling tensions between early and late-time measurements of $H_0$. Importantly, this feature sharply distinguishes our kinematic approach from alternative field-theoretic realizations of non-commutativity, such as those modifying the Klein Gordon action and energy-momentum relations. In those formulations, position dependent mass terms yield non-commutative corrections governed by an asymmetric constant tensor $\Theta^{\mu\nu}$. In a homogeneous and isotropic FLRW universe, these corrections depend strictly on $\Theta^{00}=0$, causing all non-commutative modifications to the geodesic equations, Raychaudhuri equation, and continuity equations to vanish, thereby leaving the dark energy sector untouched at leading order.
\\
\\
In the framework introduced here, $\theta^{\rm NC}$ is introduced as a free function characterizing the congruence, and we have not specified an equation of motion or microphysical origin for it, whether from an underlying quantum-gravitational structure, or some other source. Fixing this is the most immediate open problem, since it would lead to more concrete and specific predictions on the role of noncommutativity in the fundamental structure and dynamics of spacetime itself.
\section*{Acknowledgments}
We thank the participants of the Quantum Topology Seminar for the stimulating discussion and valuable feedback on this work. We are particularly grateful for the opportunity to present and discuss these results in the seminar.
\bibliographystyle{unsrt}
\bibliography{references}

@article{Carlip_2001,
   title={Quantum gravity: a progress report},
   volume={64},
   ISSN={1361-6633},
   url={http://dx.doi.org/10.1088/0034-4885/64/8/301},
   DOI={10.1088/0034-4885/64/8/301},
   number={8},
   journal={Reports on Progress in Physics},
   publisher={IOP Publishing},
   author={Carlip, S},
   year={2001},
   month=July, pages={885–942} }

@misc{rovelli1998stringsloopsotherscritical,
      title={Strings, loops and others: a critical survey of the present approaches to quantum gravity}, 
      author={Carlo Rovelli},
      year={1998},
      eprint={gr-qc/9803024},
      archivePrefix={arXiv},
      primaryClass={gr-qc},
      url={https://arxiv.org/abs/gr-qc/9803024}, 
}

@article{Blau:2009fzi,
    author = "Blau, Matthias and Theisen, Stefan",
    title = "{String theory as a theory of quantum gravity: A status report}",
    doi = "10.1007/s10714-008-0752-z",
    journal = "Gen. Rel. Grav.",
    volume = "41",
    pages = "743--755",
    year = "2009"
}

@article{Donoghue_1994,
   title={General relativity as an effective field theory: The leading quantum corrections},
   volume={50},
   ISSN={0556-2821},
   url={http://dx.doi.org/10.1103/PhysRevD.50.3874},
   DOI={10.1103/physrevd.50.3874},
   number={6},
   journal={Physical Review D},
   publisher={American Physical Society (APS)},
   author={Donoghue, John F.},
   year={1994},
   month=Sept, pages={3874–3888} }

@article{Seiberg_1999,
   title={String theory and noncommutative geometry},
   volume={1999},
   ISSN={1029-8479},
   url={http://dx.doi.org/10.1088/1126-6708/1999/09/032},
   DOI={10.1088/1126-6708/1999/09/032},
   number={09},
   journal={Journal of High Energy Physics},
   publisher={Springer Science and Business Media LLC},
   author={Seiberg, Nathan and Witten, Edward},
   year={1999},
   month=Sept, pages={032–032} }

@article{Douglas_2001,
   title={Noncommutative field theory},
   volume={73},
   ISSN={1539-0756},
   url={http://dx.doi.org/10.1103/RevModPhys.73.977},
   DOI={10.1103/revmodphys.73.977},
   number={4},
   journal={Reviews of Modern Physics},
   publisher={American Physical Society (APS)},
   author={Douglas, Michael and Nekrasov, Nikita},
   year={2001},
   month=Nov, pages={977–1029} }

@article{Szabo_2003,
   title={Quantum field theory on noncommutative spaces},
   volume={378},
   ISSN={0370-1573},
   url={http://dx.doi.org/10.1016/S0370-1573(03)00059-0},
   DOI={10.1016/s0370-1573(03)00059-0},
   number={4},
   journal={Physics Reports},
   publisher={Elsevier BV},
   author={Szabo, R},
   year={2003},
   month=May, pages={207–299} }

@article{Bayen:1977hb,
    author = "Bayen, F. and Flato, M. and Fronsdal, C. and Lichnerowicz, A. and Sternheimer, D.",
    title = "{Deformation Theory and Quantization. 2. Physical Applications}",
    reportNumber = "UCLA-77-TEP-8",
    doi = "10.1016/0003-4916(78)90225-7",
    journal = "Annals Phys.",
    volume = "111",
    pages = "111",
    year = "1978"
}

@article{Bayen:1977ha,
    author = "Bayen, F. and Flato, M. and Fronsdal, C. and Lichnerowicz, A. and Sternheimer, D.",
    title = "{Deformation Theory and Quantization. 1. Deformations of Symplectic Structures}",
    reportNumber = "UCLA/77/TEP/7",
    doi = "10.1016/0003-4916(78)90224-5",
    journal = "Annals Phys.",
    volume = "111",
    pages = "61",
    year = "1978"
}

@article{10.4310/jdg/1214455536,
author = {Boris V. Fedosov},
title = {{A simple geometrical construction of deformation quantization}},
volume = {40},
journal = {Journal of Differential Geometry},
number = {2},
publisher = {Lehigh University},
pages = {213 -- 238},
year = {1994},
doi = {10.4310/jdg/1214455536},
URL = {https://doi.org/10.4310/jdg/1214455536}
}

@article{Ishibashi_1997,
   title={A large-N reduced model as superstring},
   volume={498},
   ISSN={0550-3213},
   url={http://dx.doi.org/10.1016/S0550-3213(97)00290-3},
   DOI={10.1016/s0550-3213(97)00290-3},
   number={1-2},
   journal={Nuclear Physics B},
   publisher={Elsevier BV},
   author={Ishibashi, Nobuyuki and Kawai, Hikaru and Kitazawa, Yoshihisa and Tsuchiya, Asato},
   year={1997},
   month=Aug, pages={467–491} }

@misc{steinacker2013noncommutativegeometrymatrixmodels,
      title={Non-commutative geometry and matrix models}, 
      author={Harold Steinacker},
      year={2013},
      eprint={1109.5521},
      archivePrefix={arXiv},
      primaryClass={hep-th},
      url={https://arxiv.org/abs/1109.5521}, 
}

@book{Connes1994,
    author    = {Connes, Alain},
    title     = {Noncommutative Geometry},
    publisher = {Academic Press},
    address   = {San Diego},
    year      = {1994},
    isbn      = {9780121858605},
    pagetotal = {661}
}

@article{Chamseddine_1997,
   title={The Spectral Action Principle},
   volume={186},
   ISSN={1432-0916},
   url={http://dx.doi.org/10.1007/s002200050126},
   DOI={10.1007/s002200050126},
   number={3},
   journal={Communications in Mathematical Physics},
   publisher={Springer Science and Business Media LLC},
   author={Chamseddine, A. H. and Connes, A.},
   year={1997},
   month=July, pages={731–750} }

@article{Chamseddine_2007,
   title={Gravity and the standard model with neutrino mixing},
   volume={11},
   ISSN={1095-0753},
   url={http://dx.doi.org/10.4310/ATMP.2007.v11.n6.a3},
   DOI={10.4310/atmp.2007.v11.n6.a3},
   number={6},
   journal={Advances in Theoretical and Mathematical Physics},
   publisher={International Press of Boston},
   author={Chamseddine, Ali H. and Connes, Alain and Marcolli, Matilde},
   year={2007},
   pages={991–1089} }

@article{Groenewold:1946kp,
    author = "Groenewold, H. J.",
    title = "{On the Principles of elementary quantum mechanics}",
    doi = "10.1016/S0031-8914(46)80059-4",
    journal = "Physica",
    volume = "12",
    pages = "405--460",
    year = "1946"
}

@article{Moyal1949,
  author = {Moyal, J. E.},
  title = {Quantum mechanics as a statistical theory},
  journal = {Mathematical Proceedings of the Cambridge Philosophical Society},
  volume = {45},
  number = {1},
  pages = {99--124},
  year = {1949},
  doi = {10.1017/S030500410002482X}
}

@book{Wald1984,
  author    = {Wald, Robert M.},
  title     = {General Relativity},
  publisher = {University of Chicago Press},
  address   = {Chicago},
  year      = {1984},
  isbn      = {978-0-226-87032-8},
  url       = {https://press.uchicago.edu/ucp/books/book/chicago/G/bo5952261.html}
}

@book{HawkingEllis1973,
    author    = {Hawking, Stephen W. and Ellis, George F. R.},
    title     = {The Large Scale Structure of Space-Time},
    publisher = {Cambridge University Press},
    year      = {1973},
    series    = {Cambridge Monographs on Mathematical Physics},
    isbn      = {9780521099066},
    address   = {Cambridge}
}

@book{oneill1983semiriemannian,
  title     = {Semi-Riemannian Geometry: With Applications to Relativity},
  author    = {O'Neill, Barrett},
  series    = {Pure and Applied Mathematics},
  volume    = {103},
  year      = {1983},
  publisher = {Academic Press},
  address   = {New York},
  isbn      = {978-0-12-526740-3}
}

@book{lee2012introduction,
  title={Introduction to Smooth Manifolds},
  author={Lee, John M.},
  year={2012},
  edition={2},
  series={Graduate Texts in Mathematics},
  volume={218},
  publisher={Springer},
  address={New York, NY},
  doi={10.1007/978-1-4419-9982-5},
  isbn={978-1-4419-9981-8}
}

@book{lee2018introduction,
  author    = {Lee, John M.},
  title     = {Introduction to Riemannian Manifolds},
  edition   = {2nd},
  series    = {Graduate Texts in Mathematics},
  volume    = {176},
  publisher = {Springer International Publishing},
  address   = {Cham},
  year      = {2018},
  doi       = {10.1007/978-3-319-91755-9},
  isbn      = {978-3-319-91754-2}
}

@book{doCarmo1992Riemannian,
  author    = {do Carmo, Manfredo P.},
  title     = {Riemannian Geometry},
  publisher = {Birkh{\"a}user},
  address   = {Boston, MA},
  year      = {1992},
  translator = {Flaherty, Francis},
  series    = {Mathematics: Theory \& Applications},
  isbn      = {978-0-8176-3490-2}
}

@book{spivak1999comprehensive,
  title={A Comprehensive Introduction to Differential Geometry},
  author={Spivak, Michael},
  volume={1--5},
  edition={3rd},
  year={1999},
  publisher={Publish or Perish, Inc.},
  address={Houston, Tex.}
}

@book{warner1983foundations,
  title     = {Foundations of Differentiable Manifolds and Lie Groups},
  author    = {Warner, Frank W.},
  series    = {Graduate Texts in Mathematics},
  volume    = {94},
  year      = {1983},
  publisher = {Springer-Verlag},
  address   = {New York, NY},
  isbn      = {978-0-387-90894-6},
  url       = {https://springer.com}
}

@ARTICLE{1955PhRv...98.1123R,
       author = {{Raychaudhuri}, Amalkumar},
        title = "{Relativistic Cosmology. I}",
      journal = {Physical Review},
         year = 1955,
        month = may,
       volume = {98},
       number = {4},
        pages = {1123-1126},
          doi = {10.1103/PhysRev.98.1123},
       adsurl = {https://ui.adsabs.harvard.edu/abs/1955PhRv...98.1123R}
}

@article{frobenius1877ueber,
  author        = {Frobenius, Ferdinand Georg},
  title         = {Ueber das Pfaffsche Problem},
  journal       = {Journal f{\"u}r die reine und angewandte Mathematik},
  volume        = {1877},
  number        = {82},
  pages         = {230--315},
  year          = {1877},
  publisher     = {Walter de Gruyter, Berlin/New York},
  note          = {Published in English as "On the Pfaffian Problem"}
}

@book{mukhanov2005physical,
  title     = {Physical Foundations of Cosmology},
  author    = {Mukhanov, Viatcheslav},
  year      = {2005},
  publisher = {Cambridge University Press},
  address   = {Cambridge},
  isbn      = {9780521563987}
}

@BOOK{2003moco.book.....D,
       author = {{Dodelson}, Scott},
        title = "{Modern Cosmology}",
         year = 2003,
       adsurl = {https://ui.adsabs.harvard.edu/abs/2003moco.book.....D}
}

@book{weinberg2008cosmology,
  author    = {Steven Weinberg},
  title     = {Cosmology},
  publisher = {Oxford University Press},
  year      = {2008},
  address   = {Oxford},
  isbn      = {9780198526827}
}

@article{Perivolaropoulos_2022,
   title={Challenges for $\Lambda$CDM: An update},
   volume={95},
   ISSN={1387-6473},
   url={http://dx.doi.org/10.1016/j.newar.2022.101659},
   DOI={10.1016/j.newar.2022.101659},
   journal={New Astronomy Reviews},
   publisher={Elsevier BV},
   author={Perivolaropoulos, L. and Skara, F.},
   year={2022},
   month=Dec, pages={101659} }

@article{Di_Valentino_2021,
   title={In the realm of the Hubble tension—a review of solutions},
   volume={38},
   ISSN={1361-6382},
   url={http://dx.doi.org/10.1088/1361-6382/ac086d},
   DOI={10.1088/1361-6382/ac086d},
   number={15},
   journal={Classical and Quantum Gravity},
   publisher={IOP Publishing},
   author={Di Valentino, Eleonora and Mena, Olga and Pan, Supriya and Visinelli, Luca and Yang, Weiqiang and Melchiorri, Alessandro and Mota, David F and Riess, Adam G and Silk, Joseph},
   year={2021},
   month=July, pages={153001} }

@article{2020,
   title={Planck 2018 results: VI. Cosmological parameters},
   volume={641},
   ISSN={1432-0746},
   url={http://dx.doi.org/10.1051/0004-6361/201833910},
   DOI={10.1051/0004-6361/201833910},
   journal={Astronomy \& Astrophysics},
   publisher={EDP Sciences},
   author={Planck Collaboration},
   year={2020},
   month=Sept, pages={A6} }

@article{Perlmutter_1999,
   title={Measurements of $\Omega$ and $\Lambda$ from 42 High‐Redshift Supernovae},
   volume={517},
   ISSN={1538-4357},
   url={http://dx.doi.org/10.1086/307221},
   DOI={10.1086/307221},
   number={2},
   journal={The Astrophysical Journal},
   publisher={American Astronomical Society},
   author={Perlmutter et.al, S},
   year={1999},
   month=June, pages={565–586} }

@article{Riess_1998,
   title={Observational Evidence from Supernovae for an Accelerating Universe and a Cosmological Constant},
   volume={116},
   ISSN={0004-6256},
   url={http://dx.doi.org/10.1086/300499},
   DOI={10.1086/300499},
   number={3},
   journal={The Astronomical Journal},
   publisher={American Astronomical Society},
   author={Riess et.al, Adam G. },
   year={1998},
   month=Sept, pages={1009–1038} }

@article{Frieman_2008,
   title={Dark Energy and the Accelerating Universe},
   volume={46},
   ISSN={1545-4282},
   url={http://dx.doi.org/10.1146/annurev.astro.46.060407.145243},
   DOI={10.1146/annurev.astro.46.060407.145243},
   number={1},
   journal={Annual Review of Astronomy and Astrophysics},
   publisher={Annual Reviews},
   author={Frieman, Joshua A. and Turner, Michael S. and Huterer, Dragan},
   year={2008},
   month=Sept, pages={385–432} }

@article{Huterer_2017,
   title={Dark energy two decades after: observables, probes, consistency tests},
   volume={81},
   ISSN={1361-6633},
   url={http://dx.doi.org/10.1088/1361-6633/aa997e},
   DOI={10.1088/1361-6633/aa997e},
   number={1},
   journal={Reports on Progress in Physics},
   publisher={IOP Publishing},
   author={Huterer, Dragan and Shafer, Daniel L},
   year={2017},
   month=Dec, pages={016901} }

@article{Joyce_2015,
   title={Beyond the cosmological standard model},
   volume={568},
   ISSN={0370-1573},
   url={http://dx.doi.org/10.1016/j.physrep.2014.12.002},
   DOI={10.1016/j.physrep.2014.12.002},
   journal={Physics Reports},
   publisher={Elsevier BV},
   author={Joyce, Austin and Jain, Bhuvnesh and Khoury, Justin and Trodden, Mark},
   year={2015},
   month=Mar, pages={1–98} }

@article{Raychaudhuri1957SingularSI,
  title={Singular State in Relativistic Cosmology},
  author={Amalkumar Raychaudhuri},
  journal={Physical Review},
  year={1957},
  volume={106},
  pages={172-173},
  url={https://api.semanticscholar.org/CorpusID:120003457}
}

@article{PhysRevLett.14.57,
  title = {Gravitational Collapse and Space-Time Singularities},
  author = {Penrose, Roger},
  journal = {Phys. Rev. Lett.},
  volume = {14},
  issue = {3},
  pages = {57-59},
  numpages = {0},
  year = {1965},
  month = {Jan},
  publisher = {American Physical Society},
  doi = {10.1103/PhysRevLett.14.57},
  url = {https://link.aps.org/doi/10.1103/PhysRevLett.14.57}
}

@article{PhysRevLett.17.444,
  title = {Singularities in the Universe},
  author = {Hawking, S. W.},
  journal = {Phys. Rev. Lett.},
  volume = {17},
  issue = {8},
  pages = {444--445},
  numpages = {0},
  year = {1966},
  month = {Aug},
  publisher = {American Physical Society},
  doi = {10.1103/PhysRevLett.17.444},
  url = {https://link.aps.org/doi/10.1103/PhysRevLett.17.444}
}

@article{Hawking:1970zqf,
    author = "Hawking, S. W. and Penrose, R.",
    title = "{The Singularities of gravitational collapse and cosmology}",
    doi = "10.1098/rspa.1970.0021",
    journal = "Proc. Roy. Soc. Lond. A",
    volume = "314",
    pages = "529--548",
    year = "1970"
}

@article{Witten_2020,
   title={Light rays, singularities, and all that},
   volume={92},
   ISSN={1539-0756},
   url={http://dx.doi.org/10.1103/RevModPhys.92.045004},
   DOI={10.1103/revmodphys.92.045004},
   number={4},
   journal={Reviews of Modern Physics},
   publisher={American Physical Society (APS)},
   author={Witten, Edward},
   year={2020},
   month=Nov }

@article{Senovilla_2022,
   title={A critical appraisal of the singularity theorems},
   volume={380},
   ISSN={1471-2962},
   url={http://dx.doi.org/10.1098/rsta.2021.0174},
   DOI={10.1098/rsta.2021.0174},
   number={2222},
   journal={Philosophical Transactions of the Royal Society A: Mathematical, Physical and Engineering Sciences},
   publisher={The Royal Society},
   author={Senovilla, José M. M.},
   year={2022},
   month=Mar }

@article{Hawking:1967ju,
    author = "Hawking, Stephen",
    title = "{The occurrence of singularities in cosmology. III. Causality and singularities}",
    doi = "10.1098/rspa.1967.0164",
    journal = "Proc. Roy. Soc. Lond. A",
    volume = "300",
    pages = "187--201",
    year = "1967"
}

@book{joshi1993global,
  title={Global Aspects in Gravitation and Cosmology},
  author={Joshi, Pankaj S.},
  series={International Series of Monographs on Physics},
  volume={87},
  year={1993},
  publisher={Clarendon Press},
  address={Oxford},
  isbn={9780198500797}
}

@book{ellis2012relativistic,
  title     = {Relativistic Cosmology},
  author    = {Ellis, George F. R. and Maartens, Roy and MacCallum, Malcolm A. H.},
  publisher = {Cambridge University Press},
  year      = {2012},
  isbn      = {9781139014403}
}

@book{carroll2003spacetime,
  author    = {Carroll, Sean M.},
  title     = {Spacetime and Geometry: An Introduction to General Relativity},
  publisher = {Benjamin Cummings},
  year      = {2003},
  isbn      = {978-0805387322},
  address   = {San Francisco},
  edition   = {1st}
}

@book{MisnerThorneWheeler1973,
    title     = {Gravitation},
    author    = {Misner, Charles W. and Thorne, Kip S. and Wheeler, John Archibald},
    year      = {1973},
    publisher = {W. H. Freeman},
    address   = {San Francisco},
    isbn      = {978-0716703440}
}

@book{schutz1980geometrical,
  title     = {Geometrical Methods of Mathematical Physics},
  author    = {Schutz, Bernard F.},
  year      = {1980},
  publisher = {Cambridge University Press},
  address   = {Cambridge},
  isbn      = {9780521298872}
}

@book{lovelock1989tensors,
  title={Tensors, Differential Forms, and Variational Principles},
  author={Lovelock, David and Rund, Hanno},
  year={1989},
  publisher={Dover Publications},
  address={New York},
  isbn={978-0486658407}
}

@inbook{Curiel_2017,
   title={A Primer on Energy Conditions},
   ISBN={9781493932108},
   ISSN={2381-5841},
   url={http://dx.doi.org/10.1007/978-1-4939-3210-8_3},
   DOI={10.1007/978-1-4939-3210-8_3},
   booktitle={Towards a Theory of Spacetime Theories},
   publisher={Springer New York},
   author={Curiel, Erik},
   year={2017},
   pages={43–104} }

@book{beem2017global,
  title     = {Global Lorentzian Geometry},
  author    = {Beem, John K. and Ehrlich, Paul E. and Easley, Kevin L.},
  year      = {2017},
  edition   = {2nd},
  publisher = {CRC Press},
  address   = {Boca Raton},
  series    = {Chapman \& Hall/CRC Pure and Applied Mathematics},
  volume    = {202},
  isbn      = {9780824793241}
}

@article{Sachs1961_VI,
  author    = {Sachs, R. K.},
  title     = {Gravitational Waves in General Relativity. VI. The Outgoing Radiation Condition},
  journal   = {Proceedings of the Royal Society of London. Series A. Mathematical and Physical Sciences},
  volume    = {264},
  number    = {1318},
  pages     = {309--338},
  year      = {1961},
  publisher = {The Royal Society},
  doi       = {10.1098/rspa.1961.0202},
  url       = {https://royalsocietypublishing.org}
}

@article{Sachs:1962wk,
    author = "Sachs, R. K.",
    title = "{Gravitational waves in general relativity. VIII. Waves in asymptotically flat space-time}",
    journal = "Proceedings of the Royal Society of London. Series A. Mathematical and Physical Sciences",
    volume = "270",
    number = "1340",
    pages = "103--126",
    year = "1962",
    doi = "10.1098/rspa.1962.0206",
    url = "https://doi.org"
}

@book{penrose1972techniques,
  title     = {Techniques of Differential Topology in Relativity},
  author    = {Penrose, Roger},
  series    = {CBMS-NSF Regional Conference Series in Applied Mathematics},
  volume    = {7},
  year      = {1972},
  publisher = {Society for Industrial and Applied Mathematics},
  address   = {Philadelphia, PA},
  isbn      = {978-0-89871-005-2},
  doi       = {10.1137/1.9781611970609}
}

@article{Minguzzi2019Lorentzian,
  author    = {Minguzzi, Ettore},
  title     = {Lorentzian causality theory},
  journal   = {Living Reviews in Relativity},
  volume    = {22},
  number    = {1},
  pages     = {3},
  year      = {2019},
  doi       = {10.1007/s41114-019-0019-x},
  url       = {https://doi.org/10.1007/s41114-019-0019-x}
}

@incollection{Minguzzi2008Causal,
  author    = {Minguzzi, Ettore and S{\'a}nchez, Miguel},
  title     = {The causal hierarchy of spacetimes},
  booktitle = {Recent developments in pseudo-{R}iemannian geometry},
  series    = {ESI Lect. Math. Phys.},
  pages     = {299--358},
  publisher = {Eur. Math. Soc. Publ. House},
  address   = {Z{\"u}rich},
  year      = {2008},
  note      = {arXiv:gr-qc/0609119}
}

@article{Bernal:2004gm,
    author = "Bernal, Antonio N. and Sanchez, Miguel",
    title = "{Smoothness of time functions and the metric splitting of globally hyperbolic spacetimes}",
    journal = "Commun. Math. Phys.",
    volume = "257",
    pages = "43--50",
    year = "2005",
    eprint = "gr-qc/0401112",
    archivePrefix = "arXiv",
    primaryClass = "gr-qc",
    doi = "10.1007/s00220-005-1346-1"
}

@article{BernalSanchez2003,
  author    = {Bernal, Antonio N. and S{\'{a}}nchez, Miguel},
  title     = {On Smooth Cauchy Hypersurfaces and Geroch's Splitting Theorem},
  journal   = {Communications in Mathematical Physics},
  volume    = {243},
  number    = {3},
  pages     = {461--470},
  year      = {2003},
  doi       = {10.1007/s00220-003-0982-6},
  archivePrefix = {arXiv},
  eprint    = {gr-qc/0306108},
  primaryClass = {gr-qc}
}

@article{Bernal_2006,
  author = {Bernal, Antonio N. and S{\'{a}}nchez, Miguel},
  title = {Further results on the smoothability of Cauchy hypersurfaces and Cauchy time functions},
  journal = {Letters in Mathematical Physics},
  volume = {77},
  number = {2},
  pages = {183--197},
  year = {2006},
  month = {jul},
  publisher = {Springer Science and Business Media {LLC}},
  volume = {77},
  doi = {10.1007/s11005-006-0091-5},
  url = {https://doi.org/10.1007/s11005-006-0091-5},
  archivePrefix = {arXiv},
  eprint = {gr-qc/0512095},
  primaryClass = {gr-qc}
}

@article{Hayward:1993hg,
    author = "Hayward, Sean A.",
    title = "{General laws of black-hole dynamics}",
    eprint = "gr-qc/9303006",
    archivePrefix = "arXiv",
    primaryClass = "gr-qc",
    doi = "10.1103/PhysRevD.49.6467",
    journal = "Phys. Rev. D",
    volume = "49",
    pages = "6467--6480",
    year = "1994"
}

@article{Hayward:1997jp,
    author = "Hayward, Sean A.",
    title = "{Unified first law of black-hole dynamics and relativistic thermodynamics}",
    journal = "Class. Quant. Grav.",
    volume = "15",
    pages = "3147--3162",
    year = "1998",
    eprint = "gr-qc/9710089",
    archivePrefix = "arXiv",
    primaryClass = "gr-qc",
    doi = "10.1088/0264-9381/15/10/017"
}

@article{Hayward:2004md,
    author = "Hayward, Sean A.",
    title = "{Energy conservation for dynamical black holes}",
    journal = "Phys. Rev. Lett.",
    volume = "93",
    pages = "251101",
    year = "2004",
    eprint = "gr-qc/0404077",
    archivePrefix = "arXiv",
    primaryClass = "gr-qc",
    doi = "10.1103/PhysRevLett.93.251101"
}

@article{Galloway2000MaxPrinciples,
  author    = {Galloway, Gregory J.},
  title     = {Maximum Principles for Null Hypersurfaces and Null Splitting Theorems},
  journal   = {Annales Henri Poincar{\'e}},
  volume    = {1},
  number    = {3},
  pages     = {543--567},
  year      = {2000},
  doi       = {10.1007/s000230050006},
  archivePrefix = {arXiv},
  eprint    = {math/9909158},
  primaryClass = {math.DG}
}

@article{Beran2022Splitting,
  title={The splitting theorem for globally hyperbolic Lorentzian length spaces with non-negative timelike curvature},
  author={Beran, Tobias and Ohanyan, Argam and Rott, Felix and Solis, Didier A.},
  journal={Letters in Mathematical Physics},
  volume={113},
  number={48},
  year={2023},
  publisher={Springer},
  note={arXiv:2209.14724}
}

@article{Galloway:2005mf,
    author = "Galloway, Gregory J. and Schoen, Richard",
    title = "{A generalization of Hawking's black hole topology theorem to higher dimensions}",
    eprint = "gr-qc/0509107",
    archivePrefix = "arXiv",
    primaryClass = "gr-qc",
    doi = "10.1007/s00220-006-0019-z",
    journal = "Commun. Math. Phys.",
    volume = "266",
    pages = "571--576",
    year = "2006"
}

@article{Senovilla1998,
  author    = {Senovilla, Jos{\'e} M. M.},
  title     = {Singularity Theorems and Their Consequences},
  journal   = {General Relativity and Gravitation},
  volume    = {30},
  number    = {5},
  pages     = {701--848},
  year      = {1998},
  doi       = {10.1023/A:1018801101244},
  url       = {https://doi.org}
}

@article{Ratra:1987rm,
    author = "Ratra, Bharat and Peebles, P. J. E.",
    title = "{Cosmological Consequences of a Rolling Homogeneous Scalar Field}",
    reportNumber = "PUPT-1072",
    doi = "10.1103/PhysRevD.37.3406",
    journal = "Phys. Rev. D",
    volume = "37",
    pages = "3406",
    year = "1988"
}

@article{Wetterich_1988,
   title={Cosmology and the fate of dilatation symmetry},
   volume={302},
   ISSN={0550-3213},
   url={http://dx.doi.org/10.1016/0550-3213(88)90193-9},
   DOI={10.1016/0550-3213(88)90193-9},
   number={4},
   journal={Nuclear Physics B},
   publisher={Elsevier BV},
   author={Wetterich, C.},
   year={1988},
   month=June, pages={668–696} }

@article{Caldwell_2002,
   title={A phantom menace? Cosmological consequences of a dark energy component with super-negative equation of state},
   volume={545},
   ISSN={0370-2693},
   url={http://dx.doi.org/10.1016/S0370-2693(02)02589-3},
   DOI={10.1016/s0370-2693(02)02589-3},
   number={1-2},
   journal={Physics Letters B},
   publisher={Elsevier BV},
   author={Caldwell, R.R},
   year={2002},
   month=Oct, pages={23–29} }

@article{COPELAND_2006,
   title={DYNAMICS OF DARK ENERGY},
   volume={15},
   ISSN={1793-6594},
   url={http://dx.doi.org/10.1142/S021827180600942X},
   DOI={10.1142/s021827180600942x},
   number={11},
   journal={International Journal of Modern Physics D},
   publisher={World Scientific Pub Co Pte Ltd},
   author={COPELAND, EDMUND J. and SAMI, M. and TSUJIKAWA, SHINJI},
   year={2006},
   month=Nov, pages={1753–1935} }

@article{Adame_2025,
   title={DESI 2024 VI:  cosmological constraints from the measurements of baryon acoustic oscillations},
   volume={2025},
   ISSN={1475-7516},
   url={http://dx.doi.org/10.1088/1475-7516/2025/02/021},
   DOI={10.1088/1475-7516/2025/02/021},
   number={02},
   journal={Journal of Cosmology and Astroparticle Physics},
   publisher={IOP Publishing},
   author={DESI Collaboration},
   year={2025},
   month=Feb, pages={021} }

@article{Giar_2025,
   title={Dynamical dark energy beyond Planck? Constraints from multiple CMB probes, DESI BAO, and type-Ia supernovae},
   volume={112},
   ISSN={2470-0029},
   url={http://dx.doi.org/10.1103/ss37-cxhn},
   DOI={10.1103/ss37-cxhn},
   number={2},
   journal={Physical Review D},
   publisher={American Physical Society (APS)},
   author={Giarè, William},
   year={2025},
   month=July }

@article{Gialamas_2025,
   title={Interpreting DESI 2024 BAO: Late-time dynamical dark energy or a local effect?},
   volume={111},
   ISSN={2470-0029},
   url={http://dx.doi.org/10.1103/PhysRevD.111.043540},
   DOI={10.1103/physrevd.111.043540},
   number={4},
   journal={Physical Review D},
   publisher={American Physical Society (APS)},
   author={Gialamas, Ioannis D. and Hütsi, Gert and Kannike, Kristjan and Racioppi, Antonio and Raidal, Martti and Vasar, Martin and Veermäe, Hardi},
   year={2025},
   month=Feb }

@article{Snyder:1946qz,
    author = "Snyder, Hartland S.",
    title = "{Quantized space-time}",
    doi = "10.1103/PhysRev.71.38",
    journal = "Phys. Rev.",
    volume = "71",
    pages = "38--41",
    year = "1947"
}

@article{Matt_Gregory_2026,
   title={Geodesic Equation in Noncommutative Space: A Field Theory Perspective},
   volume={18},
   ISSN={2073-8994},
   url={http://dx.doi.org/10.3390/sym18071138},
   DOI={10.3390/sym18071138},
   number={7},
   journal={Symmetry},
   publisher={MDPI AG},
   author={Matté Gregory, Carolina and Jurić, Tajron and Pinzul, Aleksandr},
   year={2026},
   month=July, pages={1138} }
\end{document}